\documentclass[12pt,a4paper,oneside]{article}
\usepackage{a4wide}

\usepackage{amsmath} 			
\usepackage{amssymb}			
\usepackage{amsfonts}			
\usepackage{amsthm}
\usepackage{enumerate}			
\usepackage[hidelinks]{hyperref}
\hypersetup{pdflinkmargin=4pt}
\usepackage[utf8]{inputenc}		
\usepackage{tikz}
\usepackage{lscape}
\usepackage{hyperref}
\usepackage{authblk}

\usepackage{xspace}
\usetikzlibrary{matrix,calc}

\usepackage{mathtools}
\mathtoolsset{showonlyrefs}

\newcommand{\ket}[1]{|#1\rangle}
\newcommand{\bra}[1]{\langle#1|}
\newcommand{\ketbra}[2]{|#1\rangle\!\langle#2|}

\newtheorem{theorem}{Theorem}[section]
\newtheorem{proposition}{Proposition}[section]
\newtheorem{lemma}[theorem]{Lemma}

\newtheorem{definition}[theorem]{Definition}

\newcommand*{\cA}{\mathcal{A}}

\newcommand*{\cE}{\mathcal{E}}

\usepackage[noend]{algpseudocode}
\usetikzlibrary{positioning}
\DeclareMathOperator\supp{supp}

\usepackage[framemethod=tikz]{mdframed}

\newcommand{\MC}{\text{MC}}

\newcommand{\fkl}{\textsc{FKL}\xspace}
\newcommand{\hlz}{\textsc{HLZ}\xspace}

\newcommand{\tqaoa}{\textsc{QAOA}$^{+}$\xspace}
\newcommand{\tfklq}{\textsc{FKL}-\textsc{QAOA}$^{+}$\xspace}
\newcommand{\thlzq}{\textsc{HLZ}-\textsc{QAOA}$^{+}$\xspace}

\newcommand{\tqaoap}[1]{\textsc{QAOA}$_p^{+}$\xspace}
\newcommand{\tfklqp}[1]{\textsc{FKL}-\textsc{QAOA}$_{#1}^{+}$\xspace}
\newcommand{\thlzqp}[1]{\textsc{HLZ}-\textsc{QAOA}$_{#1}^{+}$\xspace}

\newcommand{\mtqaoap}[1]{\textsc{QAOA}_p^{+}\xspace}
\newcommand{\mtfklqp}[1]{\textsc{FKL}\text{-}\textsc{QAOA}_{#1}^{+}\xspace}
\newcommand{\mthlzqp}[1]{\textsc{HLZ}\text{-}\textsc{QAOA}_{#1}^{+}\xspace}

\newcommand{\qaoa}{\textsc{QAOA}\xspace}
\newcommand{\rqaoa}{\textsc{RQAOA}\xspace}
\newcommand{\qaoap}[1]{\textsc{QAOA}$_{#1}$\xspace}
\newcommand{\qaoapm}[1]{\textsc{QAOA}_{#1}\xspace}
\newcommand{\maxcut}{\text{MaxCut}\xspace}

\newcommand{\appRatio}[1]{\alpha_{#1}}

\newcommand*{\ExpE}{\mathbb{E}}
\newcommand*{\proj}[1]{|#1\rangle\langle #1|}

\newcommand{\edgeMinigraph}{\!\begin{tikzpicture}[scale=0.5, every node/.style={scale=0.5}, baseline=-1ex]
        \tikzstyle{every node}=[font=\scriptsize]
	    \coordinate (j) at (0,0);
	    \coordinate (k) at (1,0);
        \draw[thick]  (j) -- (k);
	    \node[below] at (j) {$1$};
	    \node[below] at (k) {$2$};
	    \fill (j) circle(3pt);
	    \fill (k) circle(3pt);	            
\end{tikzpicture}\!
}

\newcommand{\edgeMinigraphScript}{\!\begin{tikzpicture}[scale=0.25, every node/.style={scale=0.6}, baseline=-1ex]
	    \coordinate (j) at (0,0);
	    \coordinate (k) at (1,0);
        \draw[thick]  (j) -- (k);
	    \node[below] at (j) {$1$};
	    \node[below] at (k) {$2$};
	    \fill (j) circle(4pt);
	    \fill (k) circle(4pt);	            
\end{tikzpicture}\!
}

\newcommand{\tripletMinigraph}{\!\begin{tikzpicture}[scale=0.3, every node/.style={scale=0.3}, baseline=-1ex]
        \tikzstyle{every node}=[font=\scriptsize]
	    \coordinate (j) at (0,0);
	    \coordinate (c) at (1,0);
	    \coordinate (k) at (2,0);
        \draw[thick]  (j) -- (c) -- (k);
	    \node[below] at (j) {$j$};
	    \node[below=2pt] at (c) {$c$};
	    \node[below] at (k) {$k$};
	    \fill (j) circle(3pt);
	    \fill (c) circle(3pt);
	    \fill (k) circle(3pt);	            
\end{tikzpicture}\!
}

\newcommand{\triStarMinigraph}{\begin{tikzpicture}[scale=0.3, every node/.style={scale=0.3}, baseline=-0.6ex]
        \tikzstyle{every node}=[font=\scriptsize]
       \clip (-1.2,-1.7) rectangle (1.2,1.5);
	    \coordinate (c) at (0,0);
	    \coordinate (i) at (0,1);
	    \coordinate (k) at (0.866,-0.5);
	    \coordinate (j) at (-0.866,-0.5);
	    \draw[thick]  (j) -- (c) -- (k);
	    \draw[thick]  (c) -- (i);
	    \node[below] at (j) {$k$};
	    \node[below=2pt] at (c) {$c$};
	    \node[below] at (k) {$l$};
	    \node[right] at (i) {$j$};
	    \fill (i) circle(3pt);
	    \fill (j) circle(3pt);
	    \fill (c) circle(3pt);
	    \fill (k) circle(3pt);	            
\end{tikzpicture}
}

\newcommand{\isolatedTriangle}{\begin{tikzpicture}[scale=0.1,baseline=-0.5ex]
	    \clip (-2,-1.4) rectangle (2,2);
	    \coordinate (c) at (0,0);
	    \coordinate (i) at (0,1);
	    \coordinate (ia) at (0,2);
	    \coordinate (k) at (0.866,-0.5);
	    \coordinate (ka) at (1.732,-1);
	    \coordinate (j) at (-0.866,-0.5);
	    \coordinate (ja) at (-1.732,-1);
	    
	    \draw[semithick] (i) -- (k) -- (j) -- (i);
	    \draw[semithick] (i) -- (ia);
	    \draw[semithick] (k) -- (ka);
	    \draw[semithick] (j) -- (ja);
\end{tikzpicture}
}
\newcommand{\isolatedTriangleScript}{\begin{tikzpicture}[scale=0.08,baseline=-0.2ex, every node/.style={scale=0.5}]
	    \clip (-2,-1.4) rectangle (2,2);
	    \coordinate (c) at (0,0);
	    \coordinate (i) at (0,1);
	    \coordinate (ia) at (0,2);
	    \coordinate (k) at (0.866,-0.5);
	    \coordinate (ka) at (1.732,-1);
	    \coordinate (j) at (-0.866,-0.5);
	    \coordinate (ja) at (-1.732,-1);
	    
	    \draw (i) -- (k) -- (j) -- (i);
	    \draw (i) -- (ia);
	    \draw (k) -- (ka);
	    \draw (j) -- (ja);
\end{tikzpicture}
}

\newcommand{\TriangleScript}{\begin{tikzpicture}[scale=0.1,baseline=-0.2ex]
	    \clip (-1.2,-0.6) rectangle (1.2,1.5);
	    \coordinate (c) at (0,0);
	    \coordinate (i) at (0,1);
	    \coordinate (k) at (0.866,-0.5);
	    \coordinate (j) at (-0.866,-0.5);
	    
	    \draw (i) -- (k) -- (j) -- (i) -- (k);
\end{tikzpicture}
}
\newcommand{\crossedSquare}{\begin{tikzpicture}[scale=0.18,baseline=-0.6ex]
	    \clip (-1,-1) rectangle (1,1);
	    \coordinate (c) at (0,0);
	    \coordinate (i) at (0,0.7071);
	    \coordinate (ia) at (0,1.7071);
	    \coordinate (j) at (0,-0.7071);
	    \coordinate (ja) at (0,-1.7071);
	    \coordinate (k) at (0.7071,0);
	    \coordinate (l) at (-0.7071,0);
	    
	    \draw[semithick] (i) -- (k) -- (j) -- (l) -- (i) -- (k);
	    \draw[semithick] (k) -- (l);
	    \draw[semithick] (ia) -- (i);
	    \draw[semithick] (ja) -- (j);
\end{tikzpicture}
}
\newcommand{\crossedSquareScript}{\begin{tikzpicture}[scale=0.12,baseline=-0.4ex]
	    \clip (-1,-1) rectangle (1,1);
	    \coordinate (c) at (0,0);
	    \coordinate (i) at (0,0.7071);
	    \coordinate (ia) at (0,1.7071);
	    \coordinate (j) at (0,-0.7071);
	    \coordinate (ja) at (0,-1.7071);
	    \coordinate (k) at (0.7071,0);
	    \coordinate (l) at (-0.7071,0);
	    
	    \draw (i) -- (k) -- (j) -- (l) -- (i) -- (k);
	    \draw (k) -- (l);
	    \draw (ia) -- (i);
	    \draw (ja) -- (j);
\end{tikzpicture}
}

\usetikzlibrary{shapes.geometric}
\usetikzlibrary{shapes.misc}
\tikzset{cross/.style={cross out, draw=black, minimum size=2*(#1-\pgflinewidth), inner sep=0pt, outer sep=0pt},
	cross/.default={1pt}}

\newcommand*{\good}{\mathsf{Good}}
\newcommand*{\post}{\mathsf{Post}}
\newcommand*{\cutsize}{\mathsf{cutsize}}

\title{Twisted hybrid algorithms\\
for combinatorial optimization}
\date{}
\author[ \hspace{-1ex}]{Libor Caha\footnote{libor.caha@tum.de}}
\author[ \hspace{-1ex}]{Alexander Kliesch\footnote{kliesch@ma.tum.de}}
\author[ \hspace{-1ex}]{Robert Koenig\footnote{robert.koenig@tum.de}}
\affil[ \hspace{-1ex}]{\centering 
Zentrum Mathematik, Technical University of Munich\ \&\newline 
 Munich Center for Quantum Science and Technology, Munich, Germany
}

\begin{document}

\maketitle
\vspace{-5ex}
\begin{abstract}

Proposed hybrid algorithms encode a combinatorial cost function into a problem Hamiltonian and  optimize its energy by varying over a set of states with low circuit complexity. Classical processing  is typically only used for the choice of variational parameters following gradient descent.  As a consequence, these approaches are limited by the descriptive power of the associated  states. 

We argue that for certain combinatorial optimization problems, such algorithms can be hybridized further,  thus harnessing the power of efficient non-local classical processing. Specifically, we consider combining a  quantum variational ansatz with a greedy classical post-processing procedure for the \maxcut-problem on  $3$-regular graphs. We show that the average cut-size produced by  this method can be quantified in terms of the energy of a modified problem Hamiltonian. This motivates the consideration of an improved algorithm which variationally optimizes the energy of the modified Hamiltonian. We call this a  twisted hybrid algorithm since the additional classical processing step is combined with a different choice of variational parameters.   We exemplify the viability of this method using the quantum approximate optimization algorithm (\qaoa), giving analytic lower bounds on the expected approximation ratios achieved by twisted \qaoa. These show that the necessary non-locality of the quantum ansatz can be reduced compared to the original \qaoa: We find that for levels $p=2,\ldots, 6$, the level~$p$ can be reduced by one  while roughly maintaining the expected approximation ratio.  This  reduces the circuit depth by~$4$ and the number of variational parameters by~$2$.
  \end{abstract}

\section{Introduction}
Due to their real-world interest, problems and algorithms for combinatorial optimization figure prominently in present-day theoretical computer science. For theoretical physics, the  profound and immediate connections to the physics, e.g., of Ising or Potts models are particularly appealing. Combinatorial optimization also  provides an intriguing potential area of application of near-term quantum devices with clear figures of merit such as approximation ratios. Yet the study of quantum algorithms for these problems is still in its infancy, especially when compared to the intensely studied area of classical algorithms. For example, for classical algorithms, an established bound~\cite{MAXCUTUGC, majoritystablest} on efficiently achievable approximation ratios for \maxcut under the unique games conjecture matches that achieved by the celebrated Goemans-Williamson algorithm \cite{GW} (see also \cite{DelormePoljak93}). It appears rather unlikely that under the unique games conjecture an efficient quantum algorithm can outperform the Goemans-Williamson algorithm for generic graphs. 
 Even the more modest goal of identifying special families of instances for which a quantum algorithm outperforms comparable efficient classical algorithms appears to be out of reach. Independently of whether or not one can find a provable real-world quantum advantage in the setting of combinatorial optimization, or ends up using quantum devices as a heuristic to efficiently find approximate solutions, or finds novel classical algorithms inspired by quantum ones (as has happened before), it is natural to study to what extent existing proposals can be improved in a systematic manner with associated performance guarantees. This is what we pursue here in the context of hybrid classical-quantum algorithms.

For the problem of finding (or approximating) the maximum of a combinatorial cost function $C:\{0,1\}^n\rightarrow\mathbb{R}$ (given by polynomially many terms), typical hybrid algorithms proceed by 
defining the cost function Hamiltonian
\begin{align}
H_C&=\sum_{z\in \{0,1\}^n} C(z) \proj{z} \,
\end{align}
in terms of local terms, and a parametrized family~$\{U_{G}(\theta)\}_{\theta\in \Theta}$ of $n$-qubit unitary circuits. The later might be parametrized by the underlying graph of the cost function or in case of hardware-efficient algorithms tailored to the physical device \cite{kandala_hardware-efficient_2017}. The parametrized family give rise to 
variational ansatz states
\begin{align}
\ket{\Psi(\theta)}&=U_{G}(\theta)\ket{0}^{\otimes n}\ .
\end{align}
that can be prepared with $U_G(\theta)$ from a product state~$\ket{0}^{\otimes n}$. Measuring~$\Psi(\theta)$  in the computational basis then 
provides a sample~$z\in\{0,1\}^n$ from the distribution~$p(z)=|\langle z|\Psi(\theta)\rangle|^2$ such that the expectation value of the associated cost function is equal to the energy~$\ExpE\left[C(z)\right]=\bra{\Psi(\theta)}H_C\ket{\Psi(\theta)}$ of the state~$\Psi(\theta)$ with respect to~$H_C$. Thus the problem
of maximizing~$C$ is translated to that of finding a value of the (vector of) parameters~$\theta$ maximizing the energy of~$\Psi(\theta)$. The latter  step is envisioned to be performed e.g., by numerical gradient descent or a similar classical procedure prescribing (iteratively) what parameters~$\theta$ to try.  The computation of this prescription (according to obtained measurement results) is the classical processing part of the quantum algorithm leading to the term~{\em hybrid}. We will refer to this form of algorithm as a ``bare'' hybrid algorithm in the following.

The potential utility of this approach hinges on a number of factors. Of primary importance --
beyond questions of convergence or efficiency --   is whether the family~$\{\Psi(\theta)\}_{\theta\in\Theta}$ of states is sufficiently rich to variationally capture the (classical) correlations of  high-energy states of~$H_C$. There is an inherent tension here between the requirement of applicability using near-term devices, and the descriptive power, i.e., required complexity of these states: On the one hand, each unitary~$U_G(\theta)$ is supposed to be realized by a low-depth circuit with local gates (making it amenable to experimental realization on a near-term device), and the dimensionality of the parameter  or ``search'' space~$\Theta$ should be low to guarantee fast convergence e.g.,  of gradient descent. On the other hand, states having high energy with respect to~$H_C$ and belonging to the considered family of variational states may have intrinsically high circuit complexity, and, correspondingly, may also require a large number of variational parameters to approximate. The unavoidability of this issue has been demonstrated
using the \maxcut-problem on expander graphs with $n$~vertices and the quantum approximate optimization algorithm (\qaoa) at level~$p$: Here the parameter space is~$\Theta=[0,2\pi)^{2p}$ and the corresponding circuits~$U_G(\theta)$ have depth~$O(pd)$. Locality and symmetry of the ansatz imply that achievable expected approximation ratios are upper bounded by a constant (below that achieved by Goemans-Williamson) unless~$p=\Omega(\log n)$~\cite{Bravetal20}. In fact, the locality of the ansatz alone implies that for smaller values  of~$p$, the achieved expected approximation ratio is not better then of a random guessing for random bipartite graphs, as shown in~\cite{farhi2020quantumA}.

These fundamental limitations of ``standard'' hybrid algorithms are tied to the assumption that an increased complexity of the required quantum operations is unacceptable and/or infeasible in the near term. Under these circumstances, the only way forward appears to be to use alternative, possibly more powerful (e.g., non-local) efficient classical processing which could exploit the limited available  quantum resources more effectively. One example where a classical post-processing is used is~\cite{farhi2020quantumB}, where \qaoa is combined with a greedy ``pruning'' method to produce an  independent set of large size. Here post-processing is needed, in particular, to ensure that the output is indeed an independent set.  Another proposal in this direction is the idea of ``warm-starting'' \qaoa with a solution provided by the Goemans-Williamson algorithm~\cite{Egger2021warmstartingquantum} (see also \cite{tate2021bridging}). The warm-starting approach has the appeal that -- by construction -- the Goemans-Williamson approximation ratio can be guaranteed in this approach (assuming convergence of the energy optimization). An alternative is the recursive \qaoa (\rqaoa) method~\cite{Bravetal20,bravyicoloring} which uses \qaoa states to iteratively identify variables to eliminate. This effectively reduces the problem size but increases the connectivity and thus the circuit complexity of the iteratively obtained subproblems. Furthermore, analytical bounds on the expected approximation ratios are unknown except for very special examples~\cite{Bravetal20}. For both warm-starting \qaoa as well as \rqaoa, one deviates from the original \qaoa ansatz, leading to different variational states and corresponding quantum circuits.

\subsection*{Our contribution}
{\bf Improved hybrid algorithms.} 
Here we consider arguably more minimal adaptions of hybrid variational algorithms for the \maxcut-problem on $3$-regular graphs.  For a given bare hybrid algorithm~$\cA$ involving a family~$\{\Psi(\theta)\}_{\theta\in\Theta}$ of variational ansatz states as described above, we show how to construct a modified algorithm~$\cA^+$ which uses the {\em same} family of states~$\{\Psi(\theta)\}_{\theta\in\Theta}$.  The algorithm~$\cA^{+}$ will be called {\em twisted-$\cA$}. It requires a set of quantum operations that are  comparable (in number and complexity) to that of~$\cA$. In particular,
it involves preparing the states~$\{\Psi(\theta)\}_{\theta\in\Theta}$. In addition,~$\cA^{+}$ uses 
 extra local measurements because the hybrid optimization step is modified: the energy to be optimized is given by a modified problem Hamiltonian~$H_G^{+}$ rather than the \maxcut-problem Hamiltonian~$H_G$ associated with the considered graph~$G$. The modified Hamiltonian~$H_G^+$ is either a $3$- or $4$-local Hamiltonian and (as $H_G$) diagonal in the computational basis.  In particular, this means that measurements of up to $4$~qubits at a time in the computational basis are sufficient to determine the (expected) cost function. 
 
 By construction, the algorithms~$\cA$ and $\cA^+$ achieve (expected) cut sizes (for any fixed instance~$G$) related by the inequalities
 \begin{align}
\ExpE\left[\cutsize \left( \cA(G) \right) \right]\leq 
\ExpE\left[\cutsize \left( \cA^{+}(G) \right) \right]\ \label{eq:lowerboundscutsize}
 \end{align}
 for any (bare) hybrid algorithm~$\cA$, assuming that the optimal parameters are found in the optimization step. Indeed,~\eqref{eq:lowerboundscutsize} follows because, denoting 
 with
 \begin{align}
 \theta_*&=\arg\max_\theta \bra{\Psi(\theta)}H_G\ket{\Psi(\theta)}\ 
 \end{align}
 the optimal parameters for the Hamiltonian~$H_G$, we have  by definition of the algorithms that
 \begin{align}
 \begin{matrix}
 \ExpE\left[\cutsize \left( \cA(G) \right) \right]&=&\bra{\Psi(\theta_*)}H_G\ket{\Psi(\theta_*)}\\
   \ExpE\left[\cutsize \left( \cA^{+}(G) \right) \right]&=&\max_\theta \bra{\Psi(\theta)}H^{+}_G\ket{\Psi(\theta)} \ ,
  \end{matrix}\label{eq:expectationvaluemaxcut}
 \end{align}
 and 
 \begin{align}
 H_G^+&=H_G+\Delta_G\ ,
 \end{align}
 where $\Delta_G$ is a sum of non-negative local operators. These considerations apply to any bare hybrid algorithm~$\cA$. 
 
  {\bf Basic idea.} Our modified algorithms are directly motivated by 
   the work of Feige, Karpinski, and Langberg~\cite{FeigeKarpinskiLangberg} (referred to as \fkl in the following). These authors propose an algorithm for the \maxcut problem on $3$-regular graphs which proceeds by solving a semidefinite program relaxation (similar to Goemans and Williamson), and subsequently improving  the rounded solution by a simple greedy post-processing technique. We also consider the improvement by Halperin, Livnat, and Zwick~\cite{HalperinLivnatZwick} (referred to as \hlz below) which involves a more non-local greedy procedure. 
   
   Consider a simple motivational example of a greedy post-processing procedure that can improve a given cut.  The input will be a 3-regular graph $G=(V,E)$ and a cut $C$. We say that a vertex is unsatisfied when all three of its neighbours lie in the same partition of the cut as it does. The algorithm will repeatedly run through the vertices and check whether some of them are unsatisfied. If it finds an unsatisfied vertex it moves it to the opposite side of the cut and repeats the process with the updated cut until none of the vertices is unsatisfied. Since moving one vertex increases the cut size by 3 and potentially lowers the number of unsatisfied vertices by 4, one can show that this procedure improves the cut size by at least $\frac{3}{4}$ times the number of unsatisfied vertices in the initial cut.    Let us apply this greedy procedure to a random cut, which has an expected approximation ratio of~$1/2$. A vertex will be unsatisfied with probability $2^{-3}$. From the linearity of expectation we have that the greedy procedure will improve the cut by at least $\frac{3}{4\cdot 8}|V|$. Since $|V|=\frac{2}{3}|E|$, we achieve approximation ratio at least $\frac{1}{2}+\frac{1}{16}=0.5625$ in  expectation.

We combine these techniques with a hybrid algorithm~$\cA$ such as level-$p$ \qaoa (in the following denoted by \qaoap{p}), giving a ``twisted'' hybrid algorithm~$\cA^+$. 
The algorithm~$\cA^+$
proceeds by using the variational family of states defined by the algorithm~$\cA$ to obtain an approximate cut, but this step is modified or ``twisted'', as discussed below. The algorithm~$\cA^+$ then attempts to enlarge the cut size of the obtained cut by applying a classical post-processing procedure: We perform either the post-processing procedure by Feige, Karpinski, and Langberg (obtaining an algorithm \fkl-$\cA^+$) or the post-processing procedure by Halperin, Livnat, and Zwick (giving an algorithm \hlz-$\cA^+$). 

Let us now describe the sense in which $\cA^+$ is a ``twisted" form of~$\cA$ and not merely a hybrid algorithm augmented by a subsequent classical post-processing step. This terminology stems from the fact that in the quantum subroutine of the algorithm, the variational parameters (angles) are not optimized with respect to the original problem Hamiltonian~$H_G$. Instead, one can express the expected cut size produced by measuring a state~$\Psi(\theta)$ and using classical post-processing by the expectation value of a modified Hamiltonian~$H_G^+$ (for both~$\fkl$ and $\hlz$) in the variational state~$\Psi(\theta)$. The twisted algorithm~$\cA^+$ thus optimizes the angle~$\theta$ with respect to the modified Hamiltonian~$H_G^+$. Importantly, this does not change the ansatz/variational family of states used. This allows us to make a fair comparison (in terms of quantum resources and, especially, the number of variational parameters) to the original algorithm~$\cA$.

 {\bf Lower bounds on approximation ratios.} 
 We specialize our considerations to \qaoap{p} and establish lower bounds on the approximation ratio for bare and twisted \qaoa, i.e., we consider the algorithms \qaoap{p} and \tqaoap{p}. 
   Specifically, we consider low values of $p$ for $3$-regular graphs, triangle-free $3$-regular graphs and high girth $3$-regular graphs. 
 We denote the expected approximation ratio achieved by an algorithm~$\cA$ on a graph~$G$ with maximum cut size~$\maxcut(G)$ by 
 \begin{align}
     \appRatio{G} \left( \cA \right) := \maxcut(G)^{-1} \cdot \ExpE \left[ \cutsize \left( \cA \left( G \right) \right) \right] \ .
 \end{align}
 In the following, we will refer to the expected approximation ratio achieved by an algorithm~$\cA$ simply as the approximation of~$\cA$ (omitting the term ``expected'') unless specified otherwise. 
 In the case of $\cA = \qaoapm{p}$, $\ExpE \left[ \cutsize \left( \cA \left( G \right) \right) \right]$ is defined as in~\eqref{eq:expectationvaluemaxcut}, but with the level-$p$ \qaoa trial function~$\Psi_G(\beta,\gamma)$, $\beta,\gamma\in [0,2\pi)^{p}$ instead of $\Psi(\Theta)$. 
 
 Our results are summarized in Figure~\ref{fig:results}, which gives our lower bounds on the approximation ratio  for each of these methods.  For comparison, we also state the following known bounds on bare \qaoa for any $3$-regular graph~$G$, 
 \begin{align}
 \left.\begin{matrix}
 \appRatio{G} \left( \qaoapm{1} \right) & \geq &0.6924 \\
 \appRatio{G} \left( \qaoapm{2} \right) &\geq &0.7559 \\
 \appRatio{G} \left( \qaoapm{3} \right) &\geq & 0.79239
 \end{matrix}\right.
 \qquad 
 \begin{matrix}
 \textrm{established in~\cite{qaoaOrigPaper}}\\
 \textrm{conjectured in~\cite{qaoaOrigPaper}, established in~\cite{WurtzLove}}\\
 \textrm{conjectured in~\cite{WurtzLove}.}
 \end{matrix}
  \end{align}
  Also shown in Figure~\ref{fig:results} are the guaranteed approximation ratios of the best-known classical algorithms: This includes the Goemans-Williamson algorithm (\textsc{GW}) for general graphs (which is optimal when assuming the unique games conjecture \cite{MAXCUTUGC}) which achieves 
  \begin{align}
      \alpha_{G}(\textsc{GW})\geq 0.8785\qquad\textrm{ for any graph }~G\qquad\qquad\textrm{ (see~\cite{GW})}\ .
  \end{align}
For $3$-regular graphs, the best efficient classical algorithms are the algorithm by Feige, Karpinski, and Langberg~\cite{FeigeKarpinskiLangberg} which relies on a semidefinite program whose solution is then improved by a simple greedy post-processing technique, and a refinement of this technique by Halperin, Livnat, and Zwick~\cite{HalperinLivnatZwick}. They achieve 
\begin{align}
    \begin{matrix}
        \alpha_G \left( \fkl \right) &\geq &0.924\\
        \alpha_G \left( \hlz \right) &\geq &0.9326
    \end{matrix}
    \qquad\textrm{ for any $3$-regular graph}~G\qquad
    \begin{matrix}
        \textrm{see~\cite{FeigeKarpinskiLangberg}}\\
        \textrm{see~\cite{HalperinLivnatZwick}} .
    \end{matrix}
\end{align}

 \begin{figure}[!ht]
    \centering
    \includegraphics{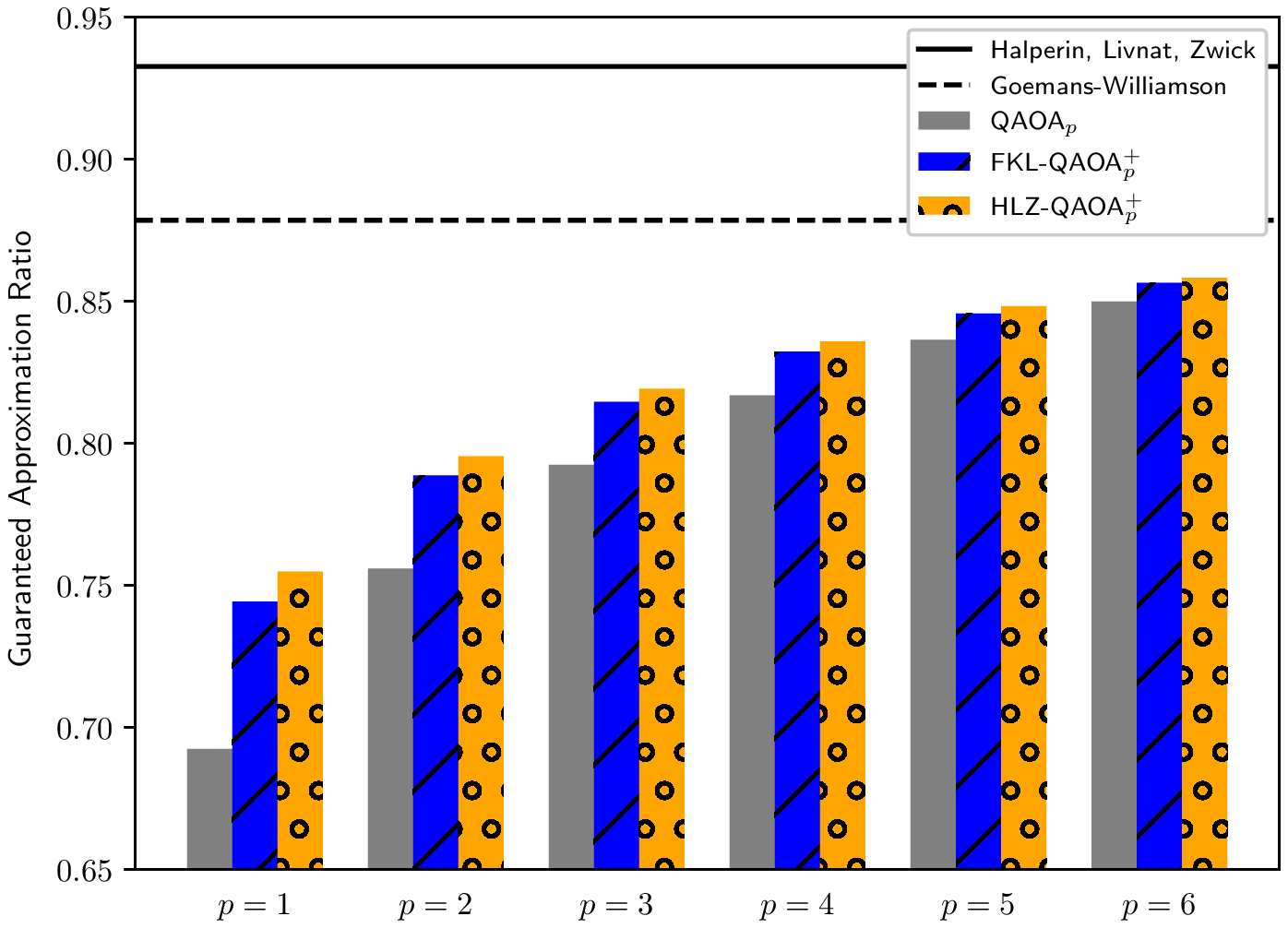}
	\begin{center}
		\begin{tabular}{l || l c | l c | c c | c | c | c}
			Method & \multicolumn{2}{c|}{$p = 1$} & \multicolumn{2}{c|}{$p = 2$} & \multicolumn{2}{c|}{$p = 3$} & $p = 4$ & $p = 5$ & $p=6$ \\ \hline \hline
			Bare \qaoap{p}~\ Eq.~\eqref{eq:QAOAboundsp45} & $\mathbf{0.6924}$ & \!\!\!\!\!\cite{qaoaOrigPaper}\!\!\!& $\mathbf{0.7559}$ & \!\!\!\!\!\cite{WurtzLove}\!\! & $0.7923$ & \!\!\!\!\!\cite{WurtzLove}\!\! & $0.8168$ & $0.8363$ & $0.8498$ \\ 
			\tfklqp{p} Prop.~\ref{prop:twFKL} & $\mathbf{0.7443}$ & & $0.7887$ & & $0.8146$ & & $0.8323$ &  $0.8457$ & $0.8564$\\
			\thlzqp{p} Prop.~\ref{prop:twHLZ} & $0.7548$ & & $0.7954$ & & $0.8191$ & &  $0.8358$ & $0.8482$ & $0.8582$ \\ 
		\end{tabular}
	\end{center}
	\caption{The main results of this work. We compare the provably guaranteed approximation ratios of bare \qaoap{p}, \tfklqp{p} and \thlzqp{p} for $3$-regular graphs with girth greater than $2p + 2$. Numbers written in  boldface  also apply to general $3$-regular graphs. 
		All quantities are rounded down to four decimals. Guaranteed approximation ratios which have been established in other work are indicated with citations. \label{fig:results}}
\end{figure}

We find that going from the original \qaoa to its twisted version leads to a significant improvement, roughly saving one level~$p$: We approximately have
\begin{align}
    \alpha \left( \textsc{QAOA}_{p - 1}^+ \right) \gtrsim \alpha \left( \textsc{QAOA}_p \right)\qquad\textrm{ for }\qquad p=2,\ldots,6\ .
\end{align}

Let us conclude by mentioning a few open problems.  One potential avenue to obtaining improved approximation ratios with hybrid algorithms is to use a different variational family of ansatz states. Here our work gives clear guidance when this is combined with classical post-processing: For a graph~$G$, the energy of a modified cost function Hamiltonian~$H_G^+=H_G+\Delta_G$ should be optimized instead of that of~$H_G$. In particular, since~$\Delta_G$ is a sum of $3$-local terms in the case of \fkl and a sum of $4$-local terms in the case of \hlz, this motivates
introducing new terms (e.g., proportional to these terms) in the ansatz.  Such a modification of the algorithm is superficially related to the fact that the classical (randomized rounding-based) algorithms of~\cite{FeigeKarpinskiLangberg,HalperinLivnatZwick} also use additional ($3$-variable) constraints in the SDP compared to the Goemans-Williamson algorithm.  We note, however, that using different variational ansatz states will require a different accounting of resources (e.g., circuit depth). In contrast, our twisted algorithms use the same circuits to prepare ansatz states as their bare version.

Another promising approach may be to combine warm-starting-type ideas with classical post-processing. Here one could consider algorithms that first solve the SDP underlying the classical algorithms~\cite{FeigeKarpinskiLangberg,HalperinLivnatZwick}, and subsequently prepare a corresponding quantum state. One may 
hope that -- similar to \cite{Egger2021warmstartingquantum} -- suitably designed approaches give a guaranteed approximation ratio matching that of these classical algorithms. 

Moving beyond combinatorial optimization problems, it is natural to ask if variational quantum algorithms for many-body quantum Hamiltonian problems (e.g., quantum analogues of \maxcut as considered in~\cite{AnshuGossetMorenz}) can be improved by similar greedy (quantum) post-processing procedures.

\subsubsection*{Outline}
In Section~\ref{sec:cpostprocessing}, we review the relevant classical post-processing methods that -- in combination with  randomized rounding of the solution of certain SDP relaxations -- yield the best known efficient classical algorithms for \maxcut on $3$-regular graphs. In Section~\ref{sec:qaoamaxcut}, we review the \qaoa and state
a few properties relevant to our subsequent analysis.  In Section~\ref{sec:twistedvariationalalg}, we motivate and define the algorithm~$\cA^+$ obtained from a hybrid algorithm~$\cA$.
Finally, in Section~\ref{sec:lowerboundsqaoa}, we establish our lower bounds on the achieved approximation ratio achieved by the twisted algorithm~\tqaoa.

\section{Classical post-processing methods for MAXCUT\label{sec:cpostprocessing}}

In this section, we describe the two classical post-processing procedures which we build on to define twisted versions of a given hybrid algorithm for the MaxCut problem on $3$-regular graphs. These post-processing procedures are subroutines
of the classical algorithms for MaxCut on bounded degree graphs and graphs with maximum degree $3$ by
Feige, Karpinski, Langberg~\cite{FeigeKarpinskiLangberg}, and Halperin, Livnat, and Zwick~\cite{HalperinLivnatZwick}, respectively. 

Recall the definition of the \maxcut problem: We are given an (undirected, simple) graph~$G=(V,E)$ and are asked assign $2$ colors to vertices $C:V\rightarrow\{0,1\}$, which we refer to as a cut of $G$, that maximizes the number~$\cutsize(C)$ of satisfied edges.
Here we say that an edge~$e=\{u,v\}$ is satisfied by~$C$ if and only if $C(u)\neq C(v)$.  The maximal size~$\cutsize(C)$ of a cut~$C$ of~$G$ is denoted~$\MC(G)$.

The Goemans-Williamson algorithm \cite{GW} for \maxcut proceeds by solving an SDP relaxation \cite{DelormePoljak93} of the \maxcut problem, and subsequently uses a randomized hyper-plane rounding to obtain a cut.
The algorithms of~\cite{FeigeKarpinskiLangberg,HalperinLivnatZwick} also proceed by first solving certain SDPs and applying randomized rounding. The obtained candidate cut is then further processed in a greedy manner in order to improve the cut size. 

Here we review these post-processing procedures and corresponding performance guarantees. One of their key features is that they can  be applied to any candidate cut~$C$ irrespective of whether it is produced e.g., by rounding the solution of an SDP, random guessing, or starting with a fixed cut.  This means that they can also be applied to the output of a hybrid algorithm. We emphasize, however, that our modified hybrid algorithms require a modification going beyond simple post-processing of the classical measurement result, see Section~\ref{sec:twistedvariationalalg} for details.

Although the guaranteed approximation ratio achieved by \hlz is better than the one achieved by \fkl, we investigate both algorithms. The reason for this lies in the locality of the procedures: while \fkl considers only the direct neighborhood of a vertex in a single step and is therefore local, \hlz also considers paths and cycles of lengths in the given graph whose lengths might potentially be unbounded and is therefore not necessarily local. We emphasize, however, that the performance of both procedures in the quantum case can be quantified by considering local operators. 

Both post-processing procedures take as input a cut~$C$. They iteratively work towards (ideally) improving the cutsize by modifying the cut. A single iteration proceeds by identifying a suitable subset~$W\subset V$ of vertices whose assigned color is flipped, i.e., replacing~$C$ by the modified cut
\begin{align}
    C^W(v):=\begin{cases}
    C(v)\qquad&\textrm{ for }\qquad v\not\in W\\
    1-C(v)\qquad&\textrm{otherwise }
    \end{cases}\ .
\end{align}

\subsection{The Feige-Karpinski-Langberg (\fkl) post-processing method}

The main idea of this post-processing step is 
the following observation: If
there are three vertices $c, j, k$ such that one of them (say, $c$) is connected to both the other ones and all three vertices are assigned the same color by the cut $C$, then flipping the value at $c$, i.e., considering $C^{\{c\}}$, will increase the size of the cut, see Figure~\ref{fig:fklmotivation}.

To formalize this, we assume that the set~$V$ of vertices of the graph $G=(V,E)$ is ordered. Without loss of generality, set $V=[n]=\{1\ldots,n\}$. The following definitions will be central:

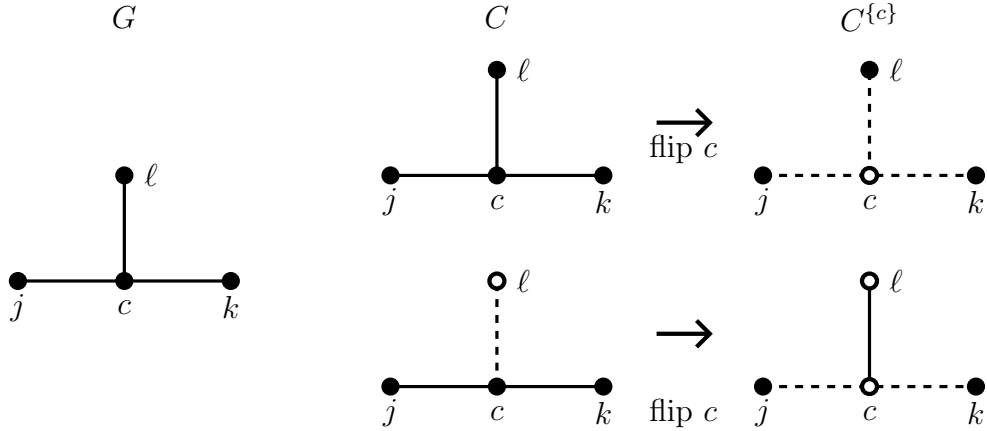
\begin{figure} 
	\centering
	\begin{tikzpicture}[scale=0.7]
	    \node at (-7, 3) {$G$};

	    \draw[very thick] (-9, -2) -- (-5, -2);
	    \draw[very thick] (-7, -2) -- (-7, 0);
	    \draw[very thick, fill] (-9, -2) circle (4pt);
	    \draw[very thick, fill] (-7, -2) circle (4pt);
	    \draw[very thick, fill] (-5, -2) circle (4pt);
	    \draw[very thick, fill] (-7, 0) circle (4pt);
	    
	    \node at (-7, -2.5) {$c$}; 
	    \node at (-9, -2.5) {$j$}; 
	    \node at (-5, -2.5) {$k$}; 
	    \node at (-6.5, 0) {$\ell$};

	    \node at (0, 3) {$C$};
	    \draw[very thick] (-2, 0) -- (2, 0);
	    \draw[very thick] (0, 0) -- (0, 2);
	    \draw[very thick, fill] (-2, 0) circle (4pt);
	    \draw[very thick, fill] (0, 0) circle (4pt);
	    \draw[very thick, fill] (2, 0) circle (4pt);
	    \draw[very thick, fill] (0, 2) circle (4pt);

	    \node at (0, -0.5) {$c$}; 
	    \node at (-2, -0.5) {$j$}; 
	    \node at (2, -0.5) {$k$}; 
	    \node at (0.5, 2) {$\ell$};
	    
	    \node at (3.5, 0.5) {flip $c$};
	    \draw[ultra thick] (3, 1) -- (4, 1);
	    \draw[ultra thick] (3.75, 0.75) -- (4, 1) -- (3.75, 1.25);

	    \draw[very thick, dashed] (5, 0) -- (9, 0);
	    \draw[very thick, dashed] (7, 0) -- (7, 2);
	    
	    \node at (7, 3) {$C^{\{c\}}$};
	    \draw[very thick, fill] (5, 0) circle (4pt);
	    \draw[very thick, white, fill] (7, 0) circle (4pt);
	    \draw[ultra thick] (7, 0) circle (4pt);
	    \draw[very thick, fill] (9, 0) circle (4pt);
	    \draw[very thick, fill] (7, 2) circle (4pt);

	    \node at (7, -0.5) {$c$}; 
	    \node at (5, -0.5) {$j$}; 
	    \node at (9, -0.5) {$k$}; 
	    \node at (7.5, 2) {$\ell$};

	    \draw[very thick] (-2, -4) -- (2, -4);
	    \draw[very thick, dashed] (0, -4) -- (0, -2);
	    \draw[very thick, fill] (-2, -4) circle (4pt);
	    \draw[very thick, fill] (0, -4) circle (4pt);
	    \draw[very thick, fill] (2, -4) circle (4pt);
	    \draw[very thick, white, fill] (0, -2) circle (4pt);
	    \draw[ultra thick] (0, -2) circle (4pt);

	    \node at (0, -4.5) {$c$}; 
	    \node at (-2, -4.5) {$j$}; 
	    \node at (2, -4.5) {$k$}; 
	    \node at (0.5, -2) {$\ell$}; 
	    
	    \node at (3.5, -4.5) {flip $c$};
	    \draw[ultra thick] (3, -3) -- (4, -3);
	    \draw[ultra thick] (3.75, -3.25) -- (4, -3) -- (3.75, -2.75);

	    \draw[very thick, dashed] (5, -4) -- (9, -4);
	    \draw[very thick] (7, -4) -- (7, -2);

	    \draw[very thick, fill] (5, -4) circle (4pt);
	    \draw[very thick, white, fill] (7, -4) circle (4pt);
	    \draw[ultra thick] (7, -4) circle (4pt);
	    \draw[very thick, fill] (9, -4) circle (4pt);
	    \draw[very thick, white, fill] (7, -2) circle (4pt);
	    \draw[ultra thick] (7, -2) circle (4pt);

	    \node at (7, -4.5) {$c$}; 
	    \node at (5, -4.5) {$j$}; 
	    \node at (9, -4.5) {$k$}; 
	    \node at (7.5, -2) {$\ell$}; 
	\end{tikzpicture}
	\caption{The main motivation behind \fkl. On the left, the closed neighborhood of a vertex $c$ is shown. Now assume that we assign a cut $C$ to $G$ and that $(c, j, k)$ is a good triplet for $C$. We distinguish two cases, depending on whether the edge $\{c, \ell\}$ is satisfied (dashed line) or unsatisfied (straight line). Top row: If $\{c, \ell\}$ is unsatisfied, flipping the value of $c$ increases the size of the cut by three (no satisfied edges are destroyed, three satisfied edges are created). Bottom row: If $\{c, \ell\}$ is satisfied, flipping the value of $c$ increases the size of the cut by one (one satisfied edge is destroyed, two satisfied edges are created).}
	\label{fig:fklmotivation}
\end{figure}

\begin{definition}[Triplets]
\begin{enumerate}[(i)]
\item
A three-tuple $(c,j,k)\in V^3$ of pairwise distinct vertices with $j<k$ is called a triplet  if $\{c, j\} \in E$ and $\{c, k\} \in E$. We call the vertex $c$ the central vertex of the triplet. The set of all triplets in~$G$ will be denoted~$T_G$.
\item
Let $C$ be a cut of $G$ and $(c,j,k)\in T_G$.
Then $(c,j,k)$ is call a good triplet for~$C$ if 
\begin{align}
			C(c) = C(j) = C(k) \ .
		\end{align}
		The set of all good triplets for~$C$ will be denoted $\good_G(C)$.
		\item
		Let $C$ be a cut of~$G$, $(c,j,k)\in \good_G(C)$ and $v\in V$. We say that $(c, j, k)$ is destroyed by flipping $v$ if $(c, j, k)$ is not a good triplet for the cut $C^{\{v\}}$. 		
		\end{enumerate}
\end{definition}

We now formulate the post-processing procedure by Feige, Karpinski, and Langberg. While the observations above show that flipping the center of a good triplet $(c, j, k)$ will increase the cutsize, we might get even better results by flipping $j$ or $k$. Furthermore, it is in our interest that the flipping does not destroy too many good triplets. Taking all this into account motivates  the procedure given in Figure~\ref{fig:fklalgorithm}.
\begin{figure}
\begin{mdframed}[
    linecolor=black,
    linewidth=2pt,
    roundcorner=4pt,
    backgroundcolor=gray!15,
    userdefinedwidth=\textwidth,
]
\begin{algorithmic}[1]
        \Function{\fkl}{$3$-regular graph $G=(V, E)$, cut $C$}
            \State $S\leftarrow \good_G \left( C \right)$
            \While{$S\not=\emptyset$}
                \State $(c,j,k)\leftarrow$ triplet $\in S$ that destroys minimal number of good triplets 
                \State $v \leftarrow \arg \max\limits_{\sigma \in \{c, j, k\}} \frac{\cutsize\left(C^{\{\sigma\}}\right) - \cutsize \left( C \right)}{\left| S \setminus \good_G \left( C^{\{\sigma\}} \right) \right|}$ 
                \State $C\leftarrow C^{\{v\}}$ 
                \State $S \leftarrow $ triplets in $S$ that are good for $C^{\{v\}}$ 
            \EndWhile
            \State \Return C
        \EndFunction
\end{algorithmic}
\end{mdframed}
\caption{The Feige, Karpinski, and Langberg improvement procedure for $3$-regular graphs \cite{FeigeKarpinskiLangberg}.\label{fig:fklalgorithm}}
\end{figure}
The following result is proven in~\cite{FeigeKarpinskiLangberg}.
\begin{lemma}[Lemma~3.2. in \cite{FeigeKarpinskiLangberg}] \label{lem:FKL}
  Let $G$ be a $3$-regular graph and let $C$ be a cut of $G$. Then the cut $C'=\fkl(G,C)$ satisfies
\begin{align}
    \cutsize(C') \geq \cutsize(C) + \frac{1}{3}|\good_G(C)|\ .
\end{align} 
\end{lemma}
Let us exemplify this improvement by using two simple examples with a $3$-regular graph~$G=(V,E)$. Consider first the trivial constant cut~$C_{\textrm{const}}$ which assigns the same color to all vertices. The cutsize of~$C_{\textrm{const}}$ is $0$, hence the approximation ratio vanishes as well, i.e.,
\begin{align}
\frac{\cutsize(C_{\textrm{const}})}{\MC(G)}=0\ .
\end{align}
Now consider the cut $C':=\fkl(G,C_{\textrm{const}})$ obtained by applying the FKL-post-processing procedure to the trivial cut. This cut achieves approximation ratio at least
\begin{align}
    \frac{\cutsize(C')}{\MC(G)}\geq 2/3\ .
    \end{align}
    This can be seen as follows: for a constant cut, every triplet is a good triplet and it is easy to see that $\left| T_G \right| = 2 \left| E \right|$ for a $3$-regular graph. Lemma~\ref{lem:FKL} then implies that the resulting cut $C$ satisfies
    $\cutsize(C') \geq \frac{2}{3} \left| E \right|$ and we obtain the claim with $\MC(G) \leq \left| E \right|$. 

As another example, consider a uniformly random cut~$C_{\textrm{random}}$ of~$G$. For such a cut, the expected 
approximation ratio  is
\begin{align}
    \ExpE\left[\frac{\cutsize(C_{\textrm{random}})}{\MC(G)}\right]=1/2\ .
\end{align}
    Let $C'':=\fkl(G, C_{\textrm{random}})$ be the result of applying the FKL-procedure to~$C_{\textrm{random}}$. Then 
    \begin{align}
    \ExpE\left[\frac{\cutsize(C'')}{\MC(G)}\right]\geq 2/3\ .
        \end{align}
To see this, note that the probability of a fixed triplet being good is equal to $\frac{1}{4}$. By linearity of expectation, we have $\mathbb{E}[|\good_G(C'')|]= \frac{1}{4} \left| T_G \right| = \frac{1}{2} \left| E \right|$. Lemma~\ref{lem:FKL} then implies that the resulting cut $C''$ satisfies $\mathbb{E}[\cutsize(C'')]\geq \left( \frac{1}{2} + \frac{1}{3} \cdot \frac{1}{2}\right) \left| E \right| = \frac{2}{3} \left| E \right| \geq \frac{2}{3} \MC(G)$.

\subsection{The Halperin-Livnat-Zwick (\hlz) post-processing method}
In 2004, Halperin, Livnat, and Zwick \cite{HalperinLivnatZwick} improved upon the algorithm of~\cite{FeigeKarpinskiLangberg}, giving an algorithm for \maxcut achieving an expected (provable) approximation ratio of at least~$0.9326$  on graphs with vertex degree at most~$3$. To the best of our knowledge\footnote{There is supposedly a slightly improved algorithm in Doror Livnat's M.Sc. thesis having an approximation ratio $0.9328$ \cite{HalperinLivnatZwick}.}, this is the best currently known efficient classical algorithm. Although their algorithm works for graphs of maximum degree 3, we will discuss a restricted and thus simpler version for triangle-free $3$-regular graphs. Unlike the \fkl-post-processing this method employs more non-local improvement procedure. The main point here is to illustrate the use of another post-processing method in the construction of twisted hybrid algorithms. We will refer to this procedure simply as  \hlz-post-processing.

Given a cut~$C$ of a triangle-free graph~$G$, this post-processing method proceeds as specified in Figure~\ref{fig:hlzalgorithm}. 
Specializing the results of~\cite{HalperinLivnatZwick} to the triangle-free case considered here gives the following statement:
\begin{figure}
\begin{mdframed}[
    linecolor=black,
    linewidth=2pt,
    roundcorner=4pt,
    backgroundcolor=gray!15,
    userdefinedwidth=\textwidth,
]
\begin{algorithmic}[1]
        \Function{\hlz}{triangle-free $3$-regular graph $G=(V, E)$, cut $C$}
            \State $V_3\leftarrow\text{vertices in $V$ with 3 unsatisfied edges by cut } C$
            \State $V_2\leftarrow\text{vertices in $V$ with 2 unsatisfied edges by cut } C$
            \While{$V_3\cup V_2\not=\emptyset$}
                \If{$V_3\not=\emptyset$}
                    \State $v\leftarrow$ vertex in $V_3$ with the smallest number of neighbours in $V_3$
                    \State $C\leftarrow C^{\{v\}}$ 
                \ElsIf{$V_2\not=\emptyset$}
                    \State $v\leftarrow$ vertex in $V_2$
                    \State $\{v_1,\ldots,v_k\}\leftarrow$ the longest path or cycle in $G[V_2]$ containing $v$
                    \State $M\leftarrow \{v_i\in\{v_1,\ldots,v_k\}\mid i \text{ is odd}\}$
                    \State $C\leftarrow C^{M}$
                \EndIf
                \State $V_3\leftarrow\text{vertices in $V$ with 3 unsatisfied edges by cut } C$
                \State $V_2\leftarrow\text{vertices in $V$ with 2 unsatisfied edges by cut } C$
            \EndWhile
            \State \textbf{return} $C$ 
        \EndFunction
\end{algorithmic}
\end{mdframed}
\caption{The Halperin, Livnat, and Zwick improvement procedure simplified to $3$-regular triangle free graphs.\label{fig:hlzalgorithm}}
\end{figure}

\begin{lemma}[Lemma 3.1. in \cite{HalperinLivnatZwick}]\label{lem:HLZ31triangleFree} Let $G$ be a $3$-regular triangle-free graph, $C$ be a cut of $G$ and $V_2$ and $V_3$ be the sets of vertices with 2 and 3 unsatisfied edges adjacent to them in the cut $C$. Then the cut $C' = \hlz(G, C)$ satisfies\footnote{Note that there is a typo in the Lemma 3.1. \cite{HalperinLivnatZwick}.}
\begin{align}
    \cutsize(C') \geq \cutsize(C) + \frac{2}{5} |V_2| + \frac{17}{15}|V_3| \ .
\end{align}
\end{lemma}

Again, let us get a feel for the impact of the procedure like we did for \fkl in certain simple scenarios, this time for a triangle-free $3$-regular graph $G = (V, E)$. Once again, consider first the trivial constant cut~$C_{\textrm{const}}$ which assigns the same color to all vertices and therefore has cutsize $0$, so the approximation ratio is $0$ as well. Considering $C' := \hlz \left( G, C_{\textrm{const}} \right)$, i.e., the cut obtained by applying the \hlz-post-processing procedure, this cut achieves an approximation ratio of at least 
\begin{align} \label{eq:hlzconstant}
    \frac{\cutsize(C')}{\MC(G)} \geq 0.7555 \ .
\end{align}
To see this, note that for a constant cut, all vertices belong to $V_3 = V$ and none to $V_2 = \emptyset$. Lemma~\ref{lem:HLZ31triangleFree} implies that $\cutsize(C') \geq \frac{17}{15} \left| V \right|$ and using that $\left| E \right| = 3 / 2 \left| V \right|\geq \MC(G)$, we obtain $\frac{\cutsize(C')}{\MC(G)}\geq\frac{\cutsize(C')}{|E|}\geq\frac{17\cdot2}{15\cdot3} \approx 0.7555$.

As another example, consider a uniformly random cut $C_\textrm{random}$ of $G$. For such a cut, the expected approximation ratio is $\frac{1}{2}$, i.e., $\ExpE \left[ \cutsize(C) \right] = \frac{1}{2} \left| E \right|$. Considering the cut $C'' := \hlz(G, C)$, the approximation ratio of this cut is 
\begin{align}
    \ExpE\left[\frac{\cutsize(C'')}{\MC(G)}\right] \geq 0.6611
\end{align}
which can be seen as follows: the probability of a vertex being in $V_3$ and $V_2$ are $2^{-3}$ and $2^{-2}$, respectively. By linearity of expectation, we have $\ExpE \left[ \left| V_3 \right| \right] = 2^{-3} \left| V \right|$ and $\ExpE \left[ \left| V_2 \right| \right] = 2^{-2} \left| V \right|$. Lemma~\ref{lem:HLZ31triangleFree} implies that $\mathbb{E}[\cutsize(C'')]\geq \frac{|E|}{2} +\frac{2}{5\cdot 4}|V|+\frac{17}{15\cdot 8}|V|$. Using that $\left| V \right| = \frac{2}{3} \left| E \right|$, we see that the approximation ratio is lower-bounded by $\frac{1}{2}+\frac{29}{180}\approx 0.6611$ in expectation value.

\section{Quantum approximate optimization and \maxcut\label{sec:qaoamaxcut}}
Here we briefly state the relevant definition for \qaoa applied to the \maxcut problem.   In Section~\ref{sec:localityuniformity}, we then discuss basic features of~\qaoa that we exploit to find lower bounds on approximation ratios. 

\subsection[Definition of the MaxCut Hamiltonian and QAOAp]{Definition of the \maxcut Hamiltonian and \qaoap{p}}
 Recall that the \maxcut problem Hamiltonian for a graph~$G=(V,E)$ is given by 
\begin{align}
    H_G = \frac{1}{2} \sum\limits_{\{u, v\} \in E} \left( I - Z_u Z_v \right)\ \label{eq:maxcuthamiltonian}
\end{align}
where a single qubit is associated with each vertex $u\in V$.
Measurement of a state~$\Psi\in(\mathbb{C}^2)^{\otimes \left| V \right|}$ in the computational basis yields a string~$C\in \{0,1\}^{\left| V \right|}$ specifying a cut~$C$ of expected size
\begin{align}
\bra{\Psi}H_G \ket{\Psi}&=\mathbb{E}\left[ \cutsize(C) \right]\ .\label{eq:maxcutexpectpsi}
\end{align}

The variational family used in \qaoa is specified by a natural number~$p$ called the level of \qaoa. For a given graph~$G=(V,E)$, the level-$p$ variational state with parameters~$(\beta,\gamma)\in [0,2\pi)^{p}\times [0,2\pi)^{p}$  is
    \begin{align}
      \ket{\psi_G( \beta , \gamma )} =
U_G(\beta,\gamma)\ket{+^{|V|}} \label{eq:qaoapwavefunction}
      \end{align}
    where $\ket{+}=\frac{1}{\sqrt{2}}(\ket{0}+\ket{1})$,  $\ket{+^{|V|}}:=\ket{+}^{\otimes |V|}$ and where
    \begin{align}
      U_G(\beta,\gamma):=\prod\limits_{m=1}^p \left[ \exp \left( -i \beta_m \sum\limits_{u\in V} X_u \right) \exp \left( -i \gamma_m H_G \right) \right]\
    \end{align}
    is the QAOA unitary. In the following, we analyze the performance of twisted algorithms derived from \qaoap{p}. 
    
\subsection{Locality and uniformity of QAOA\label{sec:localityuniformity}}
The analysis of \qaoa typically exploits its locality and uniformity, see e.g.,~\cite{qaoaOrigPaper, WurtzLove, weggemans}. Similar arguments apply to our modified versions of \qaoa. Here we state these properties in a form that will be used below to establish lower bounds on the achieved approximation ratios.

{\bf Locality of \qaoa.}    One of the defining features of this ansatz is its locality: The reduced density operator of~$\psi_G( \beta , \gamma )$ on some subset~$S\subset [n]$ of qubits is uniquely determined by~$(\beta,\gamma)$ and the ``$p$-environment''  of~$S$, a certain subgraph of~$G$. For the following analysis, it will be convenient to express this dependence in a more detailed form.

    Let $A$ be a local operator supported on a subset~$\supp(A)\subset [n]$ of qubits. Conjugation of~$A$
    by an operator of the form~$\exp(-i\beta_mX_u)$ does not change the support of~$A$ and leaves the operator invariant unless $u\in\supp(A)$. Similarly, conjugation of~$A$ by an operator of the form~$\exp(i\gamma_m Z_uZ_v)$ leaves $A$ invariant unless $\{u,v\}\cap \supp(A)\neq \emptyset$, in which case the support generically becomes~$\{u,v\}\cup \supp(A)$. Applying this reasoning iteratively shows the following: Conjugating~$A$ by the QAOA unitary~$U_G(\beta,\gamma)$  is equivalent to conjugation by a cost function unitary~$U_{G^{(p)}[\supp(A)]}(\beta,\gamma)$ associated with a subgraph~$G^{(p)}[\supp(A)]$ of~$G$. The latter is defined as follows, for any fixed subset $S\subset V$ vertices corresponding to the support of~$A$. A length-$\ell$ path starting in~$S$ is a sequence $(u_0,\ldots,u_\ell)$ of vertices such that $u_0\in S$ and $\{u_{j-1},u_{j}\}\in E$ for all $j=1,\ldots,\ell$. The subgraph~$G^{(p)}[S]$ of~$G$ is the result of taking the union of all paths of length at most~$p$ starting in~$S$. We call~$G^{(p)}[S]$ the $p$-environment of~$S$. Succinctly, this shows that
    \begin{align}
      \bra{\psi_{G}(\beta,\gamma)}A\ket{\psi_{G}(\beta,\gamma)}&=
      \bra{\psi_{G^{(p)}[\supp(A)]}(\beta,\gamma)}A\ket{\psi_{G^{(p)}[\supp(A)]}(\beta,\gamma)}\ .\label{eq:localexpressionpneighborhood}
      \end{align}
In other words, to evaluate the expectation of~$A$, it suffices to consider the \qaoa-state associated with the $p$-environment of the support of~$A$.

{\bf Uniformity of \qaoa.}
For a generic local operator~$A$ with support~$S=\supp(A)$, the quantity~$\bra{\psi_{G^{(p)}[S]}(\beta,\gamma)}A\ket{\psi_{G^{(p)}[S]}(\beta,\gamma)}$ depends on the underlying graph~$G$  only through the $p$-environment $G^{(p)}[S]$ of~$S$ and the subgraph~$G[S]$ of~$G$ induced by~$S$. In fact, for a fixed induced subgraph~$K:=G[S]$, only the equivalence class of the $p$-environment~$G^{(p)}[S]$ matters. Here two graphs $G_1$ and $G_2$ (that both contain $K$ as a subgraph) are called equivalent if and only if they are isomorphic with an isomorphism fixing~$K$. This property of \qaoa is an immediate consequence of its definition.

This motivates considering equivalence classes of $p$-environments associated with a graph~$\tilde{G}$. We denote this set by $\cE^{(p)}(\tilde{G})$ and call this the set of $p$-environments of~$\tilde{G}$. Modulo isomorphisms fixing~$\tilde{G}$, every element of $\cE^{(p)}(\tilde{G})$ is a graph that appears as a $p$-environment~$G^{(p)}[S]$
for a graph~$G$, where $S$ is a subset of vertices of~$G$ with the property that the induced subgraph is~$\tilde{G}=G[S]$. We will use individual representatives of each equivalence class to denote elements of~$\cE^{(p)}(\tilde{G})$.
For example, the set $\cE^{(1)} \big( \tripletMinigraph \big)$ is depicted in Figure~\ref{fig:tripletp1Environments}. These observations allow to reorganize expectation values that are uniform. For example, 
\begin{align}
  \bra{\psi_G(\beta,\gamma)}\left(\sum_{\{u,v\}\in E} Z_uZ_v\right)\ket{\psi_G(\beta,\gamma)}&=\sum_{\tilde{G}\in \cE^{(p)}
\big(\edgeMinigraphScript\big)}n_G(\tilde{G}) \bra{\psi_{\tilde{G}}(\beta,\gamma)}Z_1Z_2\ket{\psi_{\tilde{G}}(\beta,\gamma)}\ , \label{eq:ncatrepresent}
  \end{align}
where $n_G(\tilde{G})$ is the number of times the $p$-environment~$\tilde{G}$ appears in~$G$.

Of special interest to us will be so-called $p$-trees. Given a graph $\tilde{G}$ and $p \in \mathbb{N}$, $T^{(p)}\big( \tilde{G} \big)$ is defined as the sole tree in $\mathcal{E}^{(p)} \big( \tilde{G} \big)$, see Figures~\ref{fig:triplettrees} and \ref{fig:startrees} for examples.

\section{Twisted variational hybrid algorithms for \maxcut \label{sec:twistedvariationalalg}}
In this section, we define our twisted algorithm~$\cA^+$ given a hybrid algorithm~$\cA$. We first show in Section~\ref{sec:liftingguarantees}  that the effect of classical post-processing can be quantified in terms of the expectation value of a modified problem Hamiltonian. We then give the definition of the twisted algorithm~$\cA^+$
in Section~\ref{sec:deftwisteda}.

\subsection{Lifting performance guarantees to hybrid algorithms\label{sec:liftingguarantees}}
Lemmas~\ref{lem:FKL} and~\ref{lem:HLZ31triangleFree} provide performance guarantees for the improvement obtained by applying the (classical) \fkl- and the \hlz-algorithm to any cut~$C$. Here we show that these results easily translate to the context of hybrid algorithms.

Concretely, consider a graph~$G=(V,E)$ with $V=[n]$ and a variational ansatz state~$\Psi\in (\mathbb{C}^2)^{\otimes n}$. Measuring $\Psi$ in the computational basis provides a cut~$C\in \{0,1\}^n$ to which we can apply either the \fkl or the \hlz procedure.

Let us first consider the simpler case of \fkl, i.e., suppose that $C'=\fkl(G,C)$ is the cut obtained by applying the \fkl-post-processing to the cut~$C$. To make Lemma~\ref{lem:FKL} applicable to this setting,  we need an operator that accounts for good triplets. Such an operator  is 
\begin{align} 
  N_G:=\sum_{(c,j,k)\in T_G} \Pi
  _{c,j,k} \ , \qquad \text{where}\qquad    \Pi_{c,j,k}:=\left(\ketbra{000}{000}+\ketbra{111}{111}\right)_{c,j,k} \label{eq:goodTripletOperator}
\end{align}
with $T_G$ denoting the set of triplets in $G$. Observe that $\Pi_{c,j,k}$ is a projector onto the subspace spanned by computational basis states~$\ket{C}$ describing a
cut~$C\in \{0,1\}^n$ such that $(c, j, k)$ is a good triplet in~$C$. This implies that the expectation $\bra{\Psi}N_G\ket{\Psi}$ of $N_G$ in a state~$\Psi$ is equal to the expected number of triplets in a cut~$C$ obtained by measuring~$\Psi$ in the computational basis, i.e.,
\begin{align}
\bra{\Psi} N_G\ket{\Psi}&=\sum_{C\in \{0,1\}^n} |\langle C|\Psi\rangle|^2\cdot |\good_G(C)|=\ExpE\left[|\good_G(C)|\right]\ .\label{eq:expectedgoodsizemeasurement}
\end{align}
Correspondingly, we call $N_G$ the good triplet number operator.

Combining~\eqref{eq:expectedgoodsizemeasurement} with~\eqref{eq:maxcutexpectpsi}, we obtain the following ``quantum version'' of Lemma~\ref{lem:FKL}:
\begin{lemma}\label{lem:quantumFKL}
  Let $G=(V,E)$ be a $3$-regular graph with $V=[n]$ and $\Psi\in (\mathbb{C}^2)^{\otimes n}$. Let $C\in \{0,1\}^n$ be the result of measuring~$\Psi$ in the computational basis and $C':=\fkl(G,C)$. Then
  \begin{align}
    \mathbb{E}\left[ \cutsize(C') \right]&= \bra{\Psi} (H_G+\frac{1}{3}N_G)\ket{\Psi}\ .
  \end{align}
\end{lemma}
This lemma shows that the ``target Hamiltonian'' $H_G$  should be modified by introducing the improvement operator
\begin{align}\label{eq:impOperFKL}
    \Delta_G^{\fkl}:=\frac{1}{3}N_G \ .
\end{align}

A similar treatment applies to the \hlz-procedure. Suppose that $C' = \hlz(G, C)$ is the cut obtained by applying the \hlz-post-processing to the cut $C$. We now want to ``quantify'' Lemma~\ref{lem:HLZ31triangleFree} and therefore need two operators that account for the number of vertices with $2$ and $3$ unsatisfied edges adjacent to them, respectively. For this purpose, define
\begin{align}
    M_G^{(2)} = \sum\limits_{c \in V} \Pi^{(2)}_{c,A(c)} \qquad \text{ and } \qquad 
    M_G^{(3)} = \sum\limits_{c \in V} \Pi^{(3)}_{c,A(c)} \quad ,
\end{align}
where 
\begin{align}
    \Pi^{(2)}_{c,A(c)} :=& \sum_{b\in\{0,1\}} \ketbra{b}{b}_c \otimes P^{(b)}_{A(c)} \quad \text{ with } \quad P^{(b)}_{A(c)}:=\!\!\! \sum_{\substack{\{x,y,z\}\in \{0,1\}^3,\\ b\oplus x+b\oplus y+b\oplus z = 1}}\!\!\!\ketbra{xyz}{xyz}_{A(c)} \quad \text{ and } \\
    \Pi^{(3)}_{c,A(c)} :=&\big(\ketbra{0000}{0000}+\ketbra{1111}{1111}\big)_{\overline{A}(c)}
\end{align}
with $A(c)$ being the ordered $3$-tuple of neighbors of $c \in V$ and $\overline{A}(c)$ denoting the closed neighbourhood $\overline{A}(c) := (c, A(c)_1, A(c)_2, A(c)_3)$. Observe that $P^{(b)}_{A(c)}$ is a projector onto the sum of computational basis states that contain exactly two bits equal to $b$. Furthermore, $\Pi^{(2)}_{c,A(c)}$ is a projector onto the subspace spanned by computational basis states which are associated with exactly $2$ unsatisfied edges adjacent to $c$. Similarly, $\Pi^{(3)}_{c,A(c)}$ is a projector onto the subspace spanned by computational basis states which are associated with exactly $3$ unsatisfied edges adjacent to $c$. By abuse of notation, we use $\Pi^{(2)}_{c}$ and $\Pi^{(3)}_{c}$ whenever the graph is known from the context.

Using the same reasoning as for Lemma~\ref{lem:quantumFKL}, we obtain the following:

\begin{lemma}\label{lem:quantumHLZ}
  Let $G=(V,E)$ be a $3$-regular triangle-free graph with $V=[n]$ and $\Psi\in (\mathbb{C}^2)^{\otimes n}$. Let $C\in \{0,1\}^n$ be the result of measuring~$\Psi$ in the computational basis and $C':=\hlz(G,C)$. Then
  \begin{align}
    \mathbb{E}\left[ \cutsize(C') \right]&= \bra{\Psi} (H_G+\frac{2}{5} M_G^{(2)} + \frac{17}{15} M_G^{(3)})\ket{\Psi}\ .
  \end{align}
\end{lemma}

Therefore, $H_G$ should be modified by introducing the improvement operator 
\begin{align} \label{eq:impOperHLZ}
    \Delta_G^{\hlz} := \frac{2}{5} M_G^{(2)} + \frac{17}{15} M_G^{(3)} \ .
\end{align}

\subsection[Definition of the twisted algorithm A+]{Definition of the twisted algorithm~$\cA^+$\label{sec:deftwisteda}}
Here we present our modified variational algorithm~$\cA^+$ which we call twisted-$\cA$. We formalize a variational quantum algorithm~$\cA$ as follows: It is given by a family of states 
\begin{align}
    \cA = \left\{  \Psi_x \left( \theta \right) \right\}_{\theta \in \Theta} \ ,
\end{align}
where $x$ is an input to the algorithm, i.e., a problem instance  and $\Theta \subset \mathbb{R}^k$ for some $k \in \mathbb{N}$. Once one has chosen $\theta$, the state~$\Psi_x \left( \theta \right)$ is measured to obtain the output of the algorithm. 

In the case of the \maxcut problem, a problem instance is given by a graph~$G$. A good hybrid algorithm for this problem specifies a variational family~$\{\Psi_G(\theta)\}_{\theta\in\Theta}$ whose elements can be efficiently prepared (e.g., by a low-depth circuit) 
and which -- ideally -- contains elements with large energy (corresponding to the expected cut size) with respect to the \maxcut problem Hamiltonian $H_G$ (see Eq.~\eqref{eq:maxcuthamiltonian}).  Given such an algorithm~$\cA$, we obtain a  twisted algorithm~$\post$-$\cA^+$ 
by the following modifications, where $\post\in 
 \{\fkl,\hlz\}$ denotes the chosen classical post-processing involved (see Section~\ref{sec:cpostprocessing}):
\begin{enumerate}[(i)]
\item
In the angle optimization step, the modified cost function Hamiltonian $H_G^+=H_G+\Delta^{\post}_G$  is used. Here $\Delta_G^\fkl$ and $\Delta_G^\hlz$ are the corresponding  operators defined in Eq.~\eqref{eq:impOperFKL} and~\eqref{eq:impOperHLZ}, respectively.

\item
The classical post-processing procedure~$\post$ is applied to the measurement result obtained by measuring the optimal state.
\end{enumerate}
 Figure~\ref{fig:ppVariational} shows the general procedure.

\begin{figure}
\begin{mdframed}[
    linecolor=black,
    linewidth=2pt,
    roundcorner=4pt,
    backgroundcolor=gray!15,
    userdefinedwidth=\textwidth,
]
\begin{algorithmic}[1] 
        \Function{$\post$-$\cA^+$}{$3$-regular graph $G=(V, E)$ with $V=[n]$}
        \State Compute $\theta_* = \arg\max_{\theta\in \Theta} \bra{\Psi_G\left( \theta \right)} (H_G+\Delta_G^{\post}) \ket{\Psi_G \left( \theta \right)}$
            \State Measure $\Psi_G \left( \theta_* \right)$ in the computational basis getting outcome~$C\in\{0,1\}^n$.\label{it:steptwoppfklA}
            \State Compute $C' = \post (G, C)$.\label{it:returnppfklA}
            \State \textbf{return} $C'$
        \EndFunction
\end{algorithmic}
\end{mdframed}

\caption{The twisted algorithm~$\post$-$\cA^+$ where
$\post\in \{\fkl,\hlz\}$ and where $\cA = \left\{ \ket{\Psi_G \left( \theta \right)} \right\}_{\theta \in \Theta}$ is a variational algorithm. The measurement result~$C\in\{0,1\}^n$ obtained in step~\eqref{it:steptwoppfklA} defines a cut of~$G$.}
\label{fig:ppVariational}
\end{figure}

\section[Lower bounds on approximation ratios of twisted QAOA]{Lower bounds on approximation ratios of \tqaoa \label{sec:lowerboundsqaoa}}
Here we analyze the twisted versions of \qaoa in detail. For a graph~$G$ and $p\in\mathbb{N}$, let~$H_G$ be the Hamiltonian~\eqref{eq:maxcuthamiltonian} and $\psi_G(\beta,\gamma)$ the level-$p$ trial wavefunction defined by~\eqref{eq:qaoapwavefunction}. The twisted algorithms \tfklqp{p} and \thlzqp{p} proceed as described in Figure~\ref{fig:ppQAOA}. We prove lower bounds on the approximation ratios $\appRatio{G} \left( \mtfklqp{p} \right)$
and $\appRatio{G} \left( \mthlzqp{p} \right)$ for certain families of $3$-regular graphs~$G$. 

A remark on the proof technique is in order here: While we rely on numerical gradient descent to determine good candidate parameters, these are used to optimize our lower bounds only. In particular, the validity of the established bounds is independent of the correctness of these numerical methods. This is especially important because we consider high-dimensional optimization problems and gradient descent may or may not converge.

\begin{figure}
\begin{mdframed}[
    linecolor=black,
    linewidth=2pt,
    roundcorner=4pt,
    backgroundcolor=gray!15,
    userdefinedwidth=\textwidth,
]
\begin{algorithmic}[1]
        \Function{$\post$-$\qaoa_p$}{$3$-regular graph $G=(V, E)$ with $V=[n]$}
        \State Compute $(\beta_*,\gamma_*)=\arg\max_{(\beta,\gamma)\in [0,2\pi)^p\times [0,2\pi)^{p}} \bra{\psi_G(\beta,\gamma)}(H_G+\Delta_G^{\post})\ket{\psi_G(\beta,\gamma)}$
            \State Measure $\psi_G(\beta_*,\gamma_*)$ in the computational basis getting outcome~$C\in\{0,1\}^n$\label{it:steptwoppfklQ}
            \State Compute $C'=\post(G,C)$\label{it:returnppfklQ}
            \State \textbf{return}  $C'$
        \EndFunction
\end{algorithmic}
\end{mdframed}
\caption{The twisted algorithm $\post$-$\qaoa_p$
for $\post\in \{\fkl,\hlz\}$. \label{fig:ppQAOA}}
\end{figure}

\subsection[Approximation ratios of FKL-QAOA+ for 3-regular graphs]{Approximation ratios of \tfklq for $3$-regular graphs}
We denote the girth of a graph $G$, i.e., the size of the smallest cycle in $G$, by $g(G)$. We present two kinds of results: for $ \mtfklqp{1}$, we give a bound applicable to all $3$-regular graphs. For higher levels~$p$, we give bounds applicable to $3$-regular graphs with high girth.

\newpage
\begin{proposition} \label{prop:twFKL}
    Let $G$ be a $3$-regular graph. Then
    \begin{enumerate}[(i)]
        \item\label{it:FKLaitem} $\appRatio{G} \left( \mtfklqp{1} \right) \geq 0.7443$\ .
        \item\label{it:FKLbitem} If $g(G) \geq 7$, \,\,\! then $\appRatio{G} \left( \mtfklqp{2} \right) \geq 0.7887$\ .
        \item If $g(G) \geq 9$, \,\,\! then $\appRatio{G} \left( \mtfklqp{3} \right) \geq 0.8146$\ .
        \item If $g(G) \geq 11$, then $\appRatio{G} \left( \mtfklqp{4} \right) \geq 0.8323$\ .
        \item If $g(G) \geq 13$, then $\appRatio{G} \left( \mtfklqp{5} \right) \geq 0.8457$\ .
        \item \label{it:FKLfitem} If $g(G) \geq 15$, then $\appRatio{G} \left( \mtfklqp{6} \right) \geq 0.8564$\ .
    \end{enumerate}
\end{proposition}

\begin{proof}
        \eqref{it:FKLaitem} For brevity, let us write $\psi_G(\theta)$ for the \qaoap{1} state with parameters $\theta=(\beta,\gamma)\in [0,2\pi)^2$.
        Recall from Lemma~\ref{lem:quantumFKL} that the expected approximation ratio obtained from such a state using the \fkl-post-processing procedure is given by
        \begin{align}
          \frac{\bra{\psi_G(\theta)}\left( H_G + \Delta_G^{\fkl} \right) \ket{\psi_G(\theta)}}{\MC(G)}\ .\label{eq:appRatiop1FKL}
        \end{align}
                We follow and simplify the approach of \cite{qaoaOrigPaper,WurtzLove} and bound the ratio~\eqref{eq:appRatiop1FKL} in terms of its local contributions.
        
        We first rearrange and express the numerator of~\eqref{eq:appRatiop1FKL} as a sum over triplets. Notice that since the graph is $3$-regular, any edge lies in exactly $4$~triplets. Hence
        \begin{align}
            H_G + \Delta_G^{\fkl} = \sum_{(c,j,k)\in T_G} T_{(c, j, k)}\  \label{eq:FKLHtwisted}
        \end{align}
        where $T_{(c, j, k)}$ is the triplet operator defined as 
        \begin{align}
            T_{(c, j, k)} := \frac{H^{c,j} + H^{c,k}}{4}
            + \frac{1}{3} \Pi_{c, j, k} \qquad\textrm{ for }\qquad (c,j,k)\in T_G\label{eq:TripletOp}
        \end{align}
        and  where $H^{a, b} := \frac{1}{2} \left( I - Z_a Z_b \right)$ is term  in the MaxCut-problem Hamiltonian~$H_G$ associated with  the edge $\{a,b\}$. 
        
        Next consider the denominator in the expression~\eqref{eq:appRatiop1FKL}, i.e., the maximum  size~$\MC(G)$ of a cut. We can bound this term by the expression
        \begin{align}
            \MC(G)\leq |E|-|\isolatedTriangle(G)|-|\crossedSquare(G)|\ ,\label{eq:MCTrianglesSquares}
        \end{align}
        where  $\isolatedTriangle(G)$ is the set of isolated triangles (triangles that  share an edge with another triangle) in $G$ and $\crossedSquare(G)$ is the set of crossed squares (consisting of two triangles sharing an edge). Inequality~\eqref{eq:MCTrianglesSquares}
         follows immediately from the expression that in any cut of~$G$, there is at least one unsatisfied (i.e., ``uncut") edge in each isolated triangle because of frustration. Similarly, there is at least one unsatisfied edge in each crossed square. We note that the bound~\eqref{eq:MCTrianglesSquares} applies to any $3$-regular graph~$G$ with more than $4$~vertices because  in these graphs, any triangle is either isolated or part of a crossed squared. (Observe that for the remaining graph, the complete graph~$G=K_4$ on $4$~vertices, we have $\MC(K_4)=4$.)
         
         We can bound~$\MC(G)$ further starting from~\eqref{eq:MCTrianglesSquares} by expressing the right hand side as a sum over edges. Since every isolated triangle has three edges, we can express the number of isolated triangles as
         \begin{align}
             |\isolatedTriangle(G)|&=\frac{1}{3}\sum_{e\in E} \delta_{\isolatedTriangleScript(G)}(e)\ ,\label{it:counterone}
         \end{align}
         where $\delta_{\isolatedTriangleScript(G)}(e)$ is $1$ if the edge $e$ is part of an isolated triangle in graph the $G$ and $0$ otherwise. Similarly, we have 
         \begin{align}
                     |\crossedSquare(G)|&=\frac{1}{5}\sum_{e\in E} \delta_{\crossedSquareScript(G)}(e)\label{it:countertwo}
        \end{align}
        for crossed squares, where $\delta_{\crossedSquareScript(G)}(e)$ is $1$ if the edge~$e$ is part of a crossed square in the graph~$G$ and~$0$ otherwise.
         
        To establish our bound, we only consider the  $1$-environment of each~edge $e \in E$, i.e.,~$G^{(1)}[e]$. 
        For an edge $e\in E$ which belongs to a triangle, the $1$-environment~$G^{(1)}[e]$ is not necessarily sufficient to distinguish  whether the triangle is isolated or belongs to a crossed square: for example, this is the case for an edge~$e$ that belongs to a crossed square but is not shared by both triangles. The fraction of uncut edges (in any cut) is $1/3$~for an isolated triangle, and $1/5$~for a crossed square. Using the smaller  of these two contributions per edge, i.e., pretending that each triangle is in a crossed square, yields the bound
        \begin{align}
                  \MC(G)&     \leq\sum_{e\in E}\!\left(\!1- \frac{1}{5}\delta_{\TriangleScript(G)}(e) \!\right)\ .\quad \label{eq:edgeMCTrianglesSquares}
        \end{align}
        Here $\delta_{\TriangleScript(G)}(e)$ indicates whether the edge~$e$ is part of a triangle, i.e., $\delta_{\TriangleScript(G)}(e)$ equals~$1$ whenever the edge~$e$ is part of a triangle in graph~$G$ and~$0$ otherwise. Notice that $\delta_{\TriangleScript(G)}(e)=\delta_{\TriangleScript(G^{(1)}[e])}(e)$, therefore it is enough to examine the $1$-environments of edges to obtain the bound (the possible environments are showcased in Figure~\ref{fig:edgep1Environments}). We note that  while we have excluded $G=K_4$ in the proof of  inequality~\eqref{eq:edgeMCTrianglesSquares}, it is easy to check directly that this graph also satisfies~\eqref{eq:edgeMCTrianglesSquares}.

    Expression~\eqref{eq:edgeMCTrianglesSquares} motivates defining the local averaged \maxcut fraction of an edge~$e$ in~$G$ as 
    \begin{align}
      L^G_e:=1 - \frac{1}{5}\delta_{\TriangleScript(G)}(e)\ .
    \end{align}
    Using that every edge appears in $4$ triplets, we can reexpress the upper bound~\eqref{eq:edgeMCTrianglesSquares} as
    \begin{align}
       \MC(G)&\leq \frac{1}{4}\sum_{(c,j,k)\in T_G} \left(L^G_{\{c,j\}}+L^G_{\{c,k\}}\right)\\
       &=  \sum_{(c,j,k)\in T_G} L^G_{(c,j,k)}\ ,  \label{eq:TripletLocalMaxcut}
    \end{align}
    where
    \begin{align}
        L^G_{(c,j,k)}:=\frac{1}{4}\left(L^G_{\{c,j\}}+L^G_{\{c,k\}}\right)
    \end{align} denotes the local averaged \maxcut fraction of a triplet~$(c,j,k)\in T_G$.
    
    Inserting the upper bound~\eqref{eq:TripletLocalMaxcut}
    on~$\MC(G)$ and expression~\eqref{eq:FKLHtwisted} into~\eqref{eq:appRatiop1FKL} gives
    \begin{align}
    \frac{\bra{\psi_G(\theta)}\left( H_G + \Delta_G^{\fkl} \right) \ket{\psi_G(\theta)}}{\MC(G)}
    &\geq \frac{\sum_{(c,j,k) \in T_G} \bra{\psi_G(\theta)} T_{(c,j,k)}\ket{\psi_G(\theta)}}{\sum_{(c,j,k)\in T_G} L^G_{(c,j,k)}}\ .\label{eq:lowerboundfkla}
        \end{align}
        Recall that for any triplet~$(c,j,k)\in T_G$, the expectation value~$\bra{\psi_G(\theta)}T_{(c,j,k)}\ket{\psi_G(\theta)}$   is equal to the local expectation $\bra{\psi_{\tilde{G}}(\theta)}T_{(c,j,k)}\ket{\psi_{\tilde{G}}(\theta)}$, where~$\tilde{G}$ is the (appropriate) graph environment of the triplet. By its definition as a local quantity, the combinatorial quantity~$L^{G}_{(c,j,k)}=L^{\tilde{G}}_{(c,j,k)}$ also depends only on the corresponding graph environment.  The set of equivalence classes $\cE^{(1)}(\tripletMinigraph)=\{G_r\}_{r=1}^{11}$ of possible graph environments consists of~$11$~(equivalence classes of) graphs, see Figure~\ref{fig:tripletp1Environments}.
        Denoting -- as in~\eqref{eq:ncatrepresent} -- by
        $n_G(G_r)$ the number of times the environment~$G_r$ appears in~$G$, we can restate~\eqref{eq:lowerboundfkla} as
        \begin{align}
            \frac{\bra{\psi_G(\theta)}\left( H_G + \Delta_G^{\fkl} \right) \ket{\psi_G(\theta)}}{\MC(G)}
            &\geq \frac{\sum_{r=1}^{11}n_G(G_r) \bra{\psi_{G_r}(\theta)} T_{(c,j,k)}\ket{\psi_{G_r}(\theta)}}{\sum_{r=1}^{11} n_G(G_r)L^{G_r}_{(c,j,k)}}\  .\label{eq:lowerboundpsigmc}
        \end{align}
        Eq.~\eqref{eq:lowerboundpsigmc} is valid for any choice of~$\theta\in [0,2\pi)^2$. Suppose now that we have found some angles~$\overline{\theta}\in [0,2\pi)^2$ such that
        \begin{align}
            \frac{\bra{\psi_{G_s}(\overline{\theta})} T_{(c,j,k)}\ket{\psi_{G_s}(\overline{\theta})}}{ L^{G_s}_{(c,j,k)}}\geq 
            \frac{\bra{\psi_{G_1}(\overline{\theta})} T_{(c,j,k)}\ket{\psi_{G_1}(\overline{\theta})}}{ L^{G_1}_{(c,j,k)}}\qquad\textrm{ for all }\qquad s=2,\ldots,11\ .\label{eq:llbound}
        \end{align}
        An example of such a pair is 
        \begin{align}
            \overline{\theta}=  (\overline{\beta},\overline{\gamma})=(1.130565, 5.667705)
            \label{eq:barthetavalues}
        \end{align}as can be verified by straightforward computation.         The mediant inequality $\frac{a+b}{c+d}\geq \min\{\frac{a}{c},\frac{b}{d}\}$ implies (inductively) that
        \begin{align}
            \frac{\sum_{r=1}^{11} n_r t_r}{\sum_{r=1}^{11}n_r\ell_r}\geq \min_{r=1,\ldots,{11}} \frac{t_r}{\ell_r}
        \end{align}
        for any integers~$\{n_j\}_{j=1}^{11}\subset\mathbb{N}_0$ and non-negative scalars~$\{t_r\}_{r=1}^{11}$, $\{\ell_r\}_{r=1}^{11}$. Combining this with~\eqref{eq:llbound}, we conclude that
        \begin{align}
            \frac{\sum_{r=1}^{11}n_G(G_r) \bra{\psi_{G_r}(\theta)} T_{(c,j,k)}\ket{\psi_{G_r}(\theta)}}{\sum_{r=1}^{11} n_G(G_r)L^{G_r}_{(c,j,k)}}\geq \frac{\bra{\psi_{G_1}(\overline{\theta})} T_{(c,j,k)}\ket{\psi_{G_1}(\overline{\theta})}}{ L^{G_1}_{(c,j,k)}}\geq 0.7443\ .\label{eq:numericalboundgrtheta}
        \end{align}
        From~\eqref{eq:numericalboundgrtheta} and~\eqref{eq:lowerboundpsigmc}
        we obtain 
        \begin{align}
            \frac{\bra{\psi_G(\theta)}\left( H_G + \Delta_G^{\fkl} \right) \ket{\psi_G(\theta)}}{\MC(G)}
    &\geq 0.7443
    \end{align}
    and the claim follows by taking the maximum over~$\theta\in [0,2\pi)^2$.

    Let us briefly elaborate on the choice~\eqref{eq:barthetavalues} of parameters~$\overline{\theta}$ in this proof. 
    By direct computation, we numerically observe that the quantity
       $\max_{\theta\in [0,2\pi)^2}  
            \frac{\bra{\psi_{G_r}(\theta)} T_{(c,j,k)}\ket{\psi_{G_r}(\theta)}}{ L^{G_r}_{(c,j,k)}}$
        is minimal for $r=1$. The parameters  $\overline{\theta}\in [0,2\pi)^2$ in Eq.~\eqref{eq:barthetavalues}  are the numerically obtained angles achieving the maximum  for $r=1$. We note that their only required feature in our argument is property~\eqref{eq:barthetavalues}. This can be verified immediately. A proof that these values~$\overline{\theta}$ indeed correspond to some  maximum is not required.

    \eqref{it:FKLbitem}--\eqref{it:FKLfitem}
    Let $\psi_G(\theta)$ for $\theta \in [0,2\pi)^{2p}$ be the \qaoap{p}-wave function. 
    We again consider the expected approximation ratio given by the expression ratio~\eqref{eq:appRatiop1FKL}. We can use the
    trivial lower bound~$\MC(G)\leq |E|$ on the size of the maximum cut, giving
    \begin{align}
    \appRatio{G} \left( \mtfklqp{p} \right)\geq |E|^{-1}\cdot 
        \bra{\psi_G(\theta)}\left( H_G + \Delta_G^{\fkl} \right) \ket{\psi_G(\theta)}\label{eq:apprationGmtklqp}
    \end{align}
    for any choice of~$\theta\in [0,2\pi)^{2p}$.     The assumptions on the girth can be expressed as $g(G)>2p + 2$ for $p = 2,3,4,5,6$, i.e., the level of \qaoa. For such high-girth graphs, all relevant graph environments of an arbitrary triplet in~$G$ are isomorphic to the tree~$T^{(p)} \big( \tripletMinigraph\big)$, see Figure~\ref{fig:triplettrees}. Therefore, using~\eqref{eq:FKLHtwisted}, the bound~\eqref{eq:apprationGmtklqp} becomes 
    \begin{align}
     \appRatio{G} \left( \mtfklqp{p} \right)\geq 2\bra{\psi_{T^{(p)} \big( \tripletMinigraph\big)} \left( \theta \right)} T_{(c, j, k)} \ket{\psi_{T^{(p)} \big( \tripletMinigraph\big)} \left( \theta \right)}\label{eq:boundrhsxm}
    \end{align}
    for any choice of $\theta\in [0,2\pi)^{2p}$. We can evaluate the right hand side of this inequality using a tensor network algorithm and gradient descent to maximize the angles. In particular, in each of the cases~\eqref{it:FKLbitem}--\eqref{it:FKLfitem} we found a set of angles~$\theta$ such that the right hand side of~\eqref{eq:boundrhsxm} is equal to the value stated in the proposition. These angles are listed in Figure~\ref{fig:angles}. This completes the  proof.
\end{proof}

For sake of comparison, we also obtained the guaranteed approximation ratios of bare \qaoa for $p = 4, 5,$ and $6$ for high girth graphs. These were computed in a similar fashion as explained at the end of the proof of Proposition~\ref{prop:twFKL}:

\begin{align}
     \appRatio{G} \left( \qaoapm{p} \right)\geq \bra{\psi_{T^{(p)} ( \edgeMinigraphScript )} \left( \theta \right)} 2^{-1}(I-Z_1Z_2) \ket{\psi_{T^{(p)} ( \edgeMinigraphScript )} \left( \theta \right)}\label{eq:QAOAboundsp45}
\end{align}
The witness angles proving the lower bounds are listed in Figure~\ref{fig:angles}.

\newpage
\subsection[Approximation ratios of HLZ-QAOA+ for 3-regular graphs]{Approximation ratios of \thlzq for $3$-regular graphs}

\begin{proposition} \label{prop:twHLZ}
    Let $G=(V, E)$ be a $3$-regular graph. Then 
    \begin{enumerate}[(i)]
        \item \label{it:HLZaitem} If $G$ is triangle-free (i.e. $g(G) \geq 4$), then $\appRatio{G} \left( \mthlzqp{1} \right) \geq 0.7548$\ .
        \item \label{it:HLZbitem} If $g(G) \geq 7$, \,\,\! then $\appRatio{G} \left( \mthlzqp{2} \right) \geq 0.7954$\ .
        \item If $g(G) \geq 9$, \,\,\! then $\appRatio{G} \left( \mthlzqp{3} \right) \geq 0.8191$\ .
        \item If $g(G) \geq 11$, then $\appRatio{G} \left( \mthlzqp{4} \right) \geq 0.8358$\ .
        \item If $g(G) \geq 13$, then $\appRatio{G} \left( \mthlzqp{5} \right) \geq 0.8482$\ .
        \item \label{it:HLZfitem} If $g(G) \geq 15$, then $\appRatio{G} \left( \mtfklqp{6} \right) \geq 0.8582$\ .
    \end{enumerate}
\end{proposition}

\begin{proof}
    \eqref{it:HLZaitem} Recall from Lemma~\ref{lem:quantumHLZ} that the expected approximation ratio obtained using the \hlz-post-processing procedure is given by
        \begin{align}
            \frac{\bra{\psi_G(\theta)}\left( H_G + \Delta_G^{\hlz} \right) \ket{\psi_G(\theta)}}{\MC(G)}\ ,\label{eq:appRatiop1HLZ}
        \end{align}
        where we again use $\psi_G(\theta)$ for the \qaoap{1} state.
        
        We rearrange and express the numerator \eqref{eq:appRatiop1HLZ} as a sum over $3$-star subgraphs, as they are underlying graphs of local terms of the improvement operator $\Delta_G^{\hlz}$.
        The  $3$-star graph with the central vertex $c$ has vertices $\{c, j, k, \ell\}$ and edges $\{ \{c, j\}, \{c, k\}, \{c, \ell\}\}$ and we depict it by \triStarMinigraph. Since the graph $G$ is $3$-regular, any edge $\{a,b\}\in E$ lies in exactly $2$ stars with central vertices $a$ and $b$. Hence
        \begin{align}
            H_G+\Delta_G^{\hlz} = \sum_{c\in V} S_{c}\ , \label{eq:triStarOp}
        \end{align}
        where $S_{c}$ is the $3$-star operator
        \begin{align}
            S_{c}:=\frac{H^{c,j}+H^{c,k}+H^{c,\ell}}{2}+\frac{2}{5} \Pi^{(2)}_{c} + \frac{17}{15} \Pi^{(3)}_{c} \qquad\textrm{ for }\qquad c\in V\ , 
        \end{align}
        $(j,k,\ell)$ is the ordered neighbourhood of $c$ in $G$ and $H^{a,b}$ is again the \maxcut term on edge $\{a,b\}$.
        
        Inserting the trivial upper bound on $\MC(G)\leq |E|$ and \eqref{eq:triStarOp} into \eqref{eq:appRatiop1HLZ} gives:
        \begin{align}
            \frac{\bra{\psi_G(\theta)}\left( H_G + \Delta_G^{\hlz} \right) \ket{\psi_G(\theta)}}{\MC(G)} &\geq \frac{\sum_{c\in V} \bra{\psi_G (\theta)} S_{c} \ket{\psi_G (\theta)}}{|E|} \label{eq:HLZsumOverVertices}
        \end{align}
        We can restate \eqref{eq:HLZsumOverVertices} as a sum over the local expectation values over the graph environments
        from the set $\cE^{(1)}\big(\triStarMinigraph\big)=\{G_r\}_{r=1}^8$ (listed in Figure~\ref{fig:trifreestarp1Environments}):
        \begin{align}
            \frac{\bra{\psi_G(\theta)}\left( H_G + \Delta_G^{\hlz} \right) \ket{\psi_G(\theta)}}{\MC(G)} 
            &\geq \frac{\sum_{r=1}^8 n_G(G_r)\bra{\psi_{G_r}(\theta)} S_{c} \ket{\psi_{G_r}(\theta)}}{|E|} \ , \label{eq:HLZstarnumbers}
        \end{align}
        where $n_G(G_r)$ is number of times the environment $G_r$ appears in graph $G$.
        
        Suppose now that we have found some angles~$\overline{\theta}\in [0,2\pi)^2$ such that
        \begin{align}
            \bra{\psi_{G_s}(\overline{\theta})} S_{c}\ket{\psi_{G_s}(\overline{\theta})}
            \geq 
            \bra{\psi_{G_1}(\overline{\theta})} S_{c}\ket{\psi_{G_1}(\overline{\theta})}\qquad\textrm{ for all }\qquad s=2,\ldots,8\ .\label{eq:HLZllbound}
        \end{align}
        An example of such a pair is 
        \begin{align}
            \overline{\theta}=  (\overline{\beta},\overline{\gamma})=(0.102870, 5.669319)
            \label{eq:HLZbarthetavalues}
        \end{align}as can be verified by straightforward computation.
        
        We combine \eqref{eq:HLZstarnumbers} with \eqref{eq:HLZllbound} and use the fact that $\sum_{r=1}^8 n_{G}(G_r)=|V|=2/3|E|$ for $3$-regular graphs:
        \begin{align}
            \frac{\bra{\psi_G(\theta)}\left( H_G + \Delta_G^{\hlz} \right) \ket{\psi_G(\theta)}}{\MC(G)} \geq \frac{2}{3} \bra{\psi_{G_1}(\overline{\theta})} S_{c}\ket{\psi_{G_1}(\overline{\theta})} \geq 0.7548 \label{eq:HLZGraphTristarApprox}
        \end{align}
        and the claim follows.

    \eqref{it:HLZbitem}--\eqref{it:HLZfitem} We will follow a similar line of reasoning as in \eqref{it:HLZaitem} and Proposition~\ref{prop:twFKL}\eqref{it:FKLbitem}--\eqref{it:FKLfitem}. The assumptions again guarantee that the considered graphs are of girth greater than $2p + 2$ with $p$ being the level of \qaoa. For such high-girth graphs, all graph environments of an arbitrary star in $G$ are isomorphic to $\tilde{G} = T^{(p)} \big( \triStarMinigraph \big)$. Therefore, 
    \begin{align}
        \appRatio{G} \left( \mthlzqp{p} \right) \geq \frac{2}{3} \bra{\psi_{\tilde{G}} (\theta)} S_{c} \ket{\psi_{\tilde{G}} (\theta)} \ ,
    \end{align}
    where $\psi_{\tilde{G}}(\theta)$ for $\theta \in [0,2\pi)^{2p}$ be the \qaoap{p}-wave function. 
    We obtain witness angles by numerical optimization (listed in Figure~\ref{fig:angles}) and the claim follows.
\end{proof}
We note that the proven lower bound Proposition~\ref{prop:twHLZ}\eqref{it:HLZaitem}
on the approximation ratio~$\alpha_G(\mthlzqp{1})$
of the twisted algorithm~\qaoap{1} is below the value~$0.7555$
resulting from the application of
\hlz to a constant partition (see~\eqref{eq:hlzconstant}). An improvement over this trivial (classical) algorithm can only be observed starting from level $p\geq 2$ (cf. Proposition ~\eqref{prop:twHLZ}\eqref{it:HLZbitem}--\eqref{it:HLZfitem}). This is not surprising given the fact that the \qaoa-ansatz is very restricted, especially for small values of~$p$. In particular, 
for any angles $(\beta,\gamma)$, the  \qaoa-state $\psi_G(\beta,\gamma)$  (cf.~\eqref{eq:qaoapwavefunction}) with the usual cost function Hamiltonian~$H_G$ for \maxcut is different from both the all-zero state~$\ket{0}^{\otimes n}$ and the all-one state~$\ket{1}^{\otimes n}$. This is the case for any level~$p$ since because of the $\mathbb{Z}_2$-symmetry of the ansatz: every state~$\psi_G(\beta,\gamma)$ is an eigenstate of the operator~$X^{\otimes n}$.

\textbf{Acknowledgments.} We thank Zahra Baghali Khanian for suggesting the name ``twisted QAOA''. RK and AK acknowledge support by IBM Research. 
RK and LC gratefully acknowledge support by the European Research Council  under grant agreement no.~101001976 (project EQUIPTNT). LC thanks IBM Zurich for their hospitality.

\bibliographystyle{plain}
\bibliography{Bibliography}

\begin{thebibliography}{10}

\bibitem{AnshuGossetMorenz}
Anurag Anshu, David Gosset, and Karen Morenz.
\newblock {Beyond Product State Approximations for a Quantum Analogue of Max
  Cut}.
\newblock {\em 15th Conference on the Theory of Quantum Computation,
  Communication and Cryptography (TQC 2020)}, 158:7:1--7:15, 2020.

\bibitem{bravyicoloring}
Sergey Bravyi, Alexander Kliesch, Robert Koenig, and Eugene Tang.
\newblock Hybrid quantum-classical algorithms for approximate graph coloring,
  2020.
\newblock \href{https://arxiv.org/abs/2011.13420}{arXiv:2011.13420}, to appear
  in Quantum.

\bibitem{Bravetal20}
Sergey Bravyi, Alexander Kliesch, Robert Koenig, and Eugene Tang.
\newblock {Obstacles to Variational Quantum Optimization from Symmetry
  Protection}.
\newblock {\em Phys. Rev. Lett.}, 125:260505, Dec 2020.

\bibitem{DelormePoljak93}
Charles Delorme and Svatopluk Poljak.
\newblock Laplacian eigenvalues and the maximum cut problem.
\newblock {\em Mathematical Programming}, 62:557--574, 1993.

\bibitem{Egger2021warmstartingquantum}
Daniel~J. Egger, Jakub Mare{\v{c}}ek, and Stefan Woerner.
\newblock Warm-starting quantum optimization.
\newblock {\em {Quantum}}, 5:479, June 2021.

\bibitem{farhi2020quantumB}
Edward Farhi, David Gamarnik, and Sam Gutmann.
\newblock {The Quantum Approximate Optimization Algorithm Needs to See the
  Whole Graph: A Typical Case}, 2020.
\newblock \href{https://arxiv.org/abs/2005.09002}{arXiv:2005.09002}.

\bibitem{farhi2020quantumA}
Edward Farhi, David Gamarnik, and Sam Gutmann.
\newblock {The Quantum Approximate Optimization Algorithm Needs to See the
  Whole Graph: Worst Case Examples}, 2020.
\newblock \href{https://arxiv.org/abs/2005.08747}{arXiv:2005.08747}.

\bibitem{qaoaOrigPaper}
Edward Farhi, Jeffrey Goldstone, and Sam Gutmann.
\newblock A {Q}uantum {A}pproximate {O}ptimization {A}lgorithm, 2014.
\newblock \href{https://arxiv.org/abs/1411.4028}{arXiv:1411.4028}.

\bibitem{FeigeKarpinskiLangberg}
Uriel Feige, Marek Karpinski, and Michael Langberg.
\newblock {Improved Approximation of Max-Cut on Graphs of Bounded Degree}.
\newblock {\em J. Algorithms}, 43(2):201–219, May 2002.

\bibitem{GW}
Michel~X. Goemans and David~P. Williamson.
\newblock {Improved Approximation Algorithms for Maximum Cut and Satisfiability
  Problems Using Semidefinite Programming}.
\newblock {\em J. ACM}, 42(6):1115–1145, Nov 1995.

\bibitem{HalperinLivnatZwick}
Eran Halperin, Dror Livnat, and Uri Zwick.
\newblock {MAX CUT in cubic graphs}.
\newblock {\em Journal of Algorithms}, 53(2):169--185, 2004.

\bibitem{kandala_hardware-efficient_2017}
Abhinav Kandala, Antonio Mezzacapo, Kristan Temme, Maika Takita, Markus Brink,
  Jerry~M. Chow, and Jay~M. Gambetta.
\newblock Hardware-efficient variational quantum eigensolver for small
  molecules and quantum magnets.
\newblock {\em Nature}, 549(7671):242--246, September 2017.

\bibitem{MAXCUTUGC}
Subhash Khot, Guy Kindler, Elchanan Mossel, and Ryan O’Donnell.
\newblock {Optimal Inapproximability Results for MAX-CUT and Other 2-Variable
  CSPs?}
\newblock {\em SIAM J. Comput.}, 37(1):319–357, Apr 2007.

\bibitem{majoritystablest}
Elchanan {Mossel}, Ryan {O'Donnell}, and Krzysztof {Oleszkiewicz}.
\newblock {Noise stability of functions with low influences: invariance and
  optimality}.
\newblock {\em {Ann. Math. (2)}}, 171(1):295--341, 2010.

\bibitem{tate2021bridging}
Reuben Tate, Majid Farhadi, Creston Herold, Greg Mohler, and Swati Gupta.
\newblock {B}ridging {C}lassical and {Q}uantum with {SDP} initialized
  warm-starts for {QAOA}, 2021.
\newblock \href{https://arxiv.org/abs/2010.14021}{arXiv:2010.14021}.

\bibitem{weggemans}
Jordi~R. Weggemans, Alexander Urech, Alexander Rausch, Robert Spreeuw, Richard
  Boucherie, Florian Schreck, Kareljan Schoutens, Jiří Minář, and Florian
  Speelman.
\newblock Solving correlation clustering with {QAOA} and a {R}ydberg qudit
  system: a full-stack approach, 2021.
\newblock \href{https://arxiv.org/abs/2106.11672}{arXiv:2106.11672}.

\bibitem{WurtzLove}
Jonathan Wurtz and Peter Love.
\newblock {MaxCut quantum approximate optimization algorithm performance
  guarantees for $p>1$}.
\newblock {\em Phys. Rev. A}, 103:042612, Apr 2021.

\end{thebibliography}

\newpage
\appendix 
\section{Used graph environments}

\begin{figure}[!ht]
\begin{center}
	\begin{minipage}{0.3\textwidth}
	\centering
	\begin{tikzpicture}[scale=1]
        \clip (-2,-1.7) rectangle (2,1);
	    \coordinate (1) at (0,0);
	    \coordinate (2) at (1,0);
	    \coordinate (11) at (-0.5, 0.866); 
	    \coordinate (12) at (-0.5, -0.866); 
	    \coordinate (21) at (1.5, 0.866);
	    \coordinate (22) at (1.5, -0.866); 
	    
	    \draw[ultra thick] (1) -- (2); 
	    \node[below right= 1pt and -2pt] at (1) {$1$}; 
	    \node[below left= 1pt and -2pt] at (2) {$2$}; 
	    \fill (1) circle(2pt); 
	    \fill (2) circle(2pt); 
	    
	    \fill (11) circle(2pt); 
	    \fill (12) circle(2pt); 
	    \fill (21) circle(2pt); 
	    \fill (22) circle(2pt); 
	    
	    \draw[thick] (1) -- (11); 
	    \draw[thick] (1) -- (12); 
	    \draw[thick] (2) -- (21); 
	    \draw[thick] (2) -- (22); 
	    
	    \node at (0.5,-1.5) (label){$G_1$};
	    
	\end{tikzpicture}
	\end{minipage}
	\begin{minipage}{0.3\textwidth}
	\centering
	\begin{tikzpicture}[scale=1]
        \clip (-2,-1.7) rectangle (2,1);
	    \coordinate (1) at (0,0);
	    \coordinate (2) at (1,0);
	    \coordinate (3) at (0.5, 0.866);
	    \coordinate (12) at (-0.866, -0.5); 
	    \coordinate (22) at (1.866, -0.5); 
	    
	    \draw[ultra thick] (1) -- (2); 
	    \node[below=1pt] at (1) {$1$}; 
	    \node[below=1pt] at (2) {$2$}; 
	    \fill (1) circle(2pt); 
	    \fill (2) circle(2pt); 
	    
	    \fill (12) circle(2pt); 
	    \fill (22) circle(2pt); 
	    \fill (3) circle(2pt); 
	    
	    \draw[thick] (1) -- (12); 
	    \draw[thick] (2) -- (22);  
	    \draw[thick] (3) -- (1); 
	    \draw[thick] (3) -- (2); 
	    
	    \node at (0.5,-1.5) (label){$G_2$};
	    
	\end{tikzpicture}
	\end{minipage}
	\begin{minipage}{0.3\textwidth}
	\centering
	\begin{tikzpicture}[scale=1]
        \clip (-2,-1.7) rectangle (2,1);
	    \coordinate (1) at (0,0);
	    \coordinate (2) at (1,0);
	    \coordinate (3) at (0.5, 0.866);
	    \coordinate (4) at (0.5, -0.866); 
	    
	    \draw[ultra thick] (1) -- (2); 
	    \node[below left] at (1) {$1$}; 
	    \node[below right] at (2) {$2$}; 
	    \fill (1) circle(2pt); 
	    \fill (2) circle(2pt); 
	    
	    \fill (3) circle(2pt); 
	    \fill (4) circle(2pt); 
	    
	    \draw[thick] (1) -- (3); 
	    \draw[thick] (1) -- (4);  
	    \draw[thick] (2) -- (3); 
	    \draw[thick] (2) -- (4); 
	    
	    \node at (0.5,-1.5) (label){$G_3$};
	    
	\end{tikzpicture}
	\end{minipage}
	\end{center}
	\caption[Edge environments, $p = 1$]{The edge environments $\mathcal{E}^{(1)} \big(\edgeMinigraph\big)$ as described in Section~\ref{sec:localityuniformity}.}
	\label{fig:edgep1Environments}
\vspace{\floatsep}
\begin{center}
\begin{minipage}{0.3\textwidth}
	\centering
	\begin{tikzpicture}[scale=1]
        \clip (-2,-1.7) rectangle (2,1.5);
	    \coordinate (j) at (-1,0);
	    \coordinate (i) at (0,0);
	    \coordinate (k) at (1,0);
	    \coordinate (ia) at (0,1);
	    \coordinate (ka) at (1.5,0.866);
	    \coordinate (kb) at (1.5,-0.866);
	    \coordinate (ja) at (-1.5,0.866);
	    \coordinate (jb) at (-1.5,-0.866);
	    
	    \draw[ultra thick] (j) -- (i) -- (k);
	    \node[below=3pt] at (i) {$c$};
	    \node[below right] at (j) {$j$};
	    \node[below left] at (k) {$k$};
	    \fill (i) circle(2pt);
	    \fill (j) circle(2pt);
	    \fill (k) circle(2pt);
	    
	    \fill (ia) circle(2pt);
	    \fill (ka) circle(2pt);
	    \fill (kb) circle(2pt);
	    \fill (ja) circle(2pt);
	    \fill (jb) circle(2pt);
	    
	    \draw[thick] (i) -- (ia);
	    \draw[thick] (ka) -- (k) -- (kb);
	    \draw[thick] (ja) -- (j) -- (jb);
	    
	    \node at (0,-1.5) (label){$G_1$};
	\end{tikzpicture}
	\end{minipage}
	\begin{minipage}{0.3\textwidth}
	\centering
	\begin{tikzpicture}[scale=1]
        \clip (-2,-1.7) rectangle (2,1.5);
	    \coordinate (j) at (-1,0);
	    \coordinate (i) at (0,0);
	    \coordinate (k) at (1,0);
	    \coordinate (ia) at (0,1);
	    \coordinate (ka) at (1.5,0.866);
	    \coordinate (kb) at (1.5,-0.866);
	    \coordinate (jb) at (-1.5,-0.866);
	    
	    \draw[ultra thick] (j) -- (i) -- (k);
	    \node[below=3pt] at (i) {$c$};
	    \node[below right] at (j) {$j$};
	    \node[below left] at (k) {$k$};
	    \fill (i) circle(2pt);
	    \fill (j) circle(2pt);
	    \fill (k) circle(2pt);
	    
	    \fill (ia) circle(2pt);
	    \fill (ka) circle(2pt);
	    \fill (kb) circle(2pt);
	    \fill (jb) circle(2pt);
	    
	    \draw[thick] (i) -- (ia);
	    \draw[thick] (ka) -- (k) -- (kb);
	    \draw[thick] (ia) -- (j) -- (jb);
	    
	    \node at (0,-1.5) (label){$G_2$};
	\end{tikzpicture}
	\end{minipage}
	\begin{minipage}{0.3\textwidth}
	\centering
	\begin{tikzpicture}[scale=1]
        \clip (-2,-1.7) rectangle (2,1.5);
	    \coordinate (j) at (-1,0);
	    \coordinate (i) at (0,0);
	    \coordinate (k) at (1,0);
	    \coordinate (ia) at (0,1);
	    \coordinate (kb) at (1.5,-0.866);
	    \coordinate (jb) at (-1.5,-0.866);
	    
	    \draw[ultra thick] (j) -- (i) -- (k);
	    \node[below=3pt] at (i) {$c$};
	    \node[below right] at (j) {$j$};
	    \node[below left] at (k) {$k$};
	    \fill (i) circle(2pt);
	    \fill (j) circle(2pt);
	    \fill (k) circle(2pt);
	    
	    \fill (ia) circle(2pt);
	    \fill (kb) circle(2pt);
	    \fill (jb) circle(2pt);
	    
	    \draw[thick] (i) -- (ia);
	    \draw[thick] (ia) -- (k) -- (kb);
	    \draw[thick] (ia) -- (j) -- (jb);
	    
	    \node at (0,-1.5) (label){$G_3$};
	\end{tikzpicture}
	\end{minipage}
	\\ 
	\begin{minipage}{0.3\textwidth}
	\centering
	\begin{tikzpicture}[scale=1]
        \clip (-2,-1.7) rectangle (2,1.5);
	    \coordinate (j) at (-1,0);
	    \coordinate (i) at (0,0);
	    \coordinate (k) at (1,0);
	    \coordinate (ia) at (0,1);
	    \coordinate (b) at (0, -1);
	    
	    \draw[ultra thick] (j) -- (i) -- (k);
	    \node[below=3pt] at (i) {$c$};
	    \node[below] at (j) {$j$};
	    \node[below] at (k) {$k$};
	    \fill (i) circle(2pt);
	    \fill (j) circle(2pt);
	    \fill (k) circle(2pt);
	    
	    \fill (ia) circle(2pt);
	    \fill (b) circle(2pt);
	    
	    \draw[thick] (i) -- (ia);
	    \draw[thick] (ia) -- (k) -- (b);
	    \draw[thick] (ia) -- (j) -- (b);
	    
	    \node at (0,-1.5) (label){$G_4$};
	\end{tikzpicture}
	\end{minipage} 
	\begin{minipage}{0.3\textwidth}
	\centering
	\begin{tikzpicture}[scale=1]
        \clip (-2,-1.7) rectangle (2,1.5);
	    \coordinate (j) at (-1,0);
	    \coordinate (i) at (0,0);
	    \coordinate (k) at (1,0);
	    \coordinate (ia) at (0,1);
	    \coordinate (ib) at (0,0.5);
	    \coordinate (b) at (0, -1);
	    
	    \draw[ultra thick] (j) -- (i) -- (k);
	    \node[below=3pt] at (i) {$c$};
	    \node[below] at (j) {$j$};
	    \node[below] at (k) {$k$};
	    \fill (i) circle(2pt);
	    \fill (j) circle(2pt);
	    \fill (k) circle(2pt);
	    
	    \fill (ia) circle(2pt);
	    \fill (ib) circle(2pt);
	    \fill (b) circle(2pt);
	    
	    \draw[thick] (i) -- (ib);
	    \draw[thick] (ia) -- (k) -- (b);
	    \draw[thick] (ia) -- (j) -- (b);
	    
	    \node at (0,-1.5) (label){$G_5$};
	\end{tikzpicture}
	\end{minipage} 
	\begin{minipage}{0.3\textwidth}
	\centering
	\begin{tikzpicture}[scale=1]
        \clip (-2,-1.7) rectangle (2,1.5);
	    \coordinate (j) at (-1,0);
	    \coordinate (i) at (0,0);
	    \coordinate (k) at (1,0);
	    \coordinate (ia) at (0,1);
	    \coordinate (ka) at (1.5, 0.866); 
	    \coordinate (b) at (0, -1);
	    
	    \draw[ultra thick] (j) -- (i) -- (k);
	    \node[below=3pt] at (i) {$c$};
	    \node[below] at (j) {$j$};
	    \node[below] at (k) {$k$};
	    \fill (i) circle(2pt);
	    \fill (j) circle(2pt);
	    \fill (k) circle(2pt);
	    
	    \fill (ia) circle(2pt);
	    \fill (ka) circle(2pt);
	    \fill (b) circle(2pt);
	    
	    \draw[thick] (i) -- (ia);
	    \draw[thick] (ka) -- (k) -- (b);
	    \draw[thick] (ia) -- (j) -- (b);
	    
	    \node at (0,-1.5) (label){$G_6$};
	\end{tikzpicture}
	\end{minipage} 
	\\ 
	\begin{minipage}{0.3\textwidth}
	\centering
	\begin{tikzpicture}[scale=1]
        \clip (-2,-1.7) rectangle (2,1.5);
	    \coordinate (j) at (-1,0);
	    \coordinate (i) at (0,0);
	    \coordinate (k) at (1,0);
	    \coordinate (ia) at (0,1);
	    \coordinate (ka) at (1.5,0.866);
	    \coordinate (ja) at (-1.5,0.866);
	    \coordinate (b) at (0, -1); 
	    
	    \draw[ultra thick] (j) -- (i) -- (k);
	    \node[below=3pt] at (i) {$c$};
	    \node[below] at (j) {$j$};
	    \node[below] at (k) {$k$};
	    \fill (i) circle(2pt);
	    \fill (j) circle(2pt);
	    \fill (k) circle(2pt);
	    
	    \fill (ia) circle(2pt);
	    \fill (ka) circle(2pt);
	    \fill (ja) circle(2pt);
	    \fill (b) circle(2pt);
	    
	    \draw[thick] (i) -- (ia);
	    \draw[thick] (ka) -- (k);
	    \draw[thick] (ja) -- (j);
	    \draw[thick] (j) -- (b) -- (k);
	    
	    \node at (0,-1.5) (label){$G_7$};
	\end{tikzpicture}
	\end{minipage} 
	\begin{minipage}{0.3\textwidth}
	\centering
	\begin{tikzpicture}[scale=1]
        \clip (-2,-1.7) rectangle (2,1.5);
	    \coordinate (j) at (-1,0);
	    \coordinate (i) at (0,0);
	    \coordinate (k) at (1,0);
	    \coordinate (ia) at (0,1);
	    \coordinate (ka) at (1.5,0.866);
	    \coordinate (ja) at (-1.5,0.866);
	    
	    \draw[ultra thick] (j) -- (i) -- (k);
	    \node[below=3pt] at (i) {$c$};
	    \node[below left] at (j) {$j$};
	    \node[below right] at (k) {$k$};
	    \fill (i) circle(2pt);
	    \fill (j) circle(2pt);
	    \fill (k) circle(2pt);
	    
	    \fill (ia) circle(2pt);
	    \fill (ka) circle(2pt);
	    \fill (ja) circle(2pt);
	    
	    \draw[thick] (i) -- (ia);
	    \draw[thick] (ka) -- (k);
	    \draw[thick] (ja) -- (j);
	    \draw[thick] (1,0) arc (0:-180:1);
	    
	    \node at (0,-1.5) (label){$G_8$};
	\end{tikzpicture}
	\end{minipage} 
	\begin{minipage}{0.3\textwidth}
	\centering
	\begin{tikzpicture}[scale=1]
        \clip (-2,-1.7) rectangle (2,1.5);
	    \coordinate (j) at (-1,0);
	    \coordinate (i) at (0,0);
	    \coordinate (k) at (1,0);
	    \coordinate (ia) at (0,1);
	    \coordinate (ka) at (1.5, 0.866); 
	    
	    \draw[ultra thick] (j) -- (i) -- (k);
	    \node[below=3pt] at (i) {$c$};
	    \node[below left] at (j) {$j$};
	    \node[below right] at (k) {$k$};
	    \fill (i) circle(2pt);
	    \fill (j) circle(2pt);
	    \fill (k) circle(2pt);
	    
	    \fill (ia) circle(2pt);
	    \fill (ka) circle(2pt);
	    
	    \draw[thick] (i) -- (ia);
	    \draw[thick] (ka) -- (k);
	    \draw[thick] (ia) -- (j);
	    \draw[thick] (1,0) arc (0:-180:1);
	    
	    \node at (0,-1.5) (label){$G_9$};
	\end{tikzpicture}
	\end{minipage} 
	\\ 
	\begin{minipage}{0.3\textwidth}
	\centering
	\begin{tikzpicture}[scale=1]
        \clip (-2,-1.7) rectangle (2,1.5);
	    \coordinate (j) at (-1,0);
	    \coordinate (i) at (0,0);
	    \coordinate (k) at (1,0);
	    \coordinate (ia) at (0,1);
	    
	    \draw[ultra thick] (j) -- (i) -- (k);
	    \node[below=3pt] at (i) {$c$};
	    \node[below left] at (j) {$j$};
	    \node[below right] at (k) {$k$};
	    \fill (i) circle(2pt);
	    \fill (j) circle(2pt);
	    \fill (k) circle(2pt);
	    
	    \fill (ia) circle(2pt);
	    
	    \draw[thick] (i) -- (ia);
	    \draw[thick] (ia) -- (k);
	    \draw[thick] (ia) -- (j);
	    \draw[thick] (1,0) arc (0:-180:1);
	    
	    \node at (0,-1.5) (label){$G_{10}$};
	\end{tikzpicture}
	\end{minipage}
	\begin{minipage}{0.3\textwidth}
	\centering
	\begin{tikzpicture}[scale=1]
        \clip (-2,-1.7) rectangle (2,1.5);
	    \coordinate (j) at (-1,0);
	    \coordinate (i) at (0,0);
	    \coordinate (k) at (1,0);
	    \coordinate (ia) at (0,1);
	    \coordinate (ib) at (0,0.5);
	    
	    \draw[ultra thick] (j) -- (i) -- (k);
	    \node[below=3pt] at (i) {$c$};
	    \node[below left] at (j) {$j$};
	    \node[below right] at (k) {$k$};
	    \fill (i) circle(2pt);
	    \fill (j) circle(2pt);
	    \fill (k) circle(2pt);
	    
	    \fill (ia) circle(2pt);
	    \fill (ib) circle(2pt);
	    
	    \draw[thick] (i) -- (ib);
	    \draw[thick] (ia) -- (k);
	    \draw[thick] (ia) -- (j);
	    
	    \draw[thick] (1,0) arc (0:-180:1);
	    
	    \node at (0,-1.5) (label){$G_{11}$};
	\end{tikzpicture}
	\end{minipage}
	\end{center}
	\caption[Triplet environments, $p = 1$]{The triplet environments $\mathcal{E}^{(1)} \big(\tripletMinigraph \big)$ as described in Section~\ref{sec:localityuniformity}.}
	\label{fig:tripletp1Environments}
\end{figure}

\begin{figure}[!ht]
\begin{center}
	\begin{minipage}{0.3\textwidth}
	\centering
	\begin{tikzpicture}[scale=1]
        \clip (-2,-2.5) rectangle (2,2.2);
	    \coordinate (c) at (0,0);
	    \coordinate (i) at (0,1);
	    \coordinate (k) at (0.866,-0.5);
	    \coordinate (ia) at (0.866,1.5);
	    \coordinate (ib) at (-0.866,1.5);
	    \coordinate (j) at (-0.866,-0.5);
	    \coordinate (ja) at (-1.732,0);
	    \coordinate (jb) at (-0.866,-1.5);
	    \coordinate (ka) at (1.732,0);
	    \coordinate (kb) at (0.866,-1.5);

	    \draw[ultra thick] (c) -- (i);
	    \draw[ultra thick] (c) -- (j);
	    \draw[ultra thick] (c) -- (k);
	    \node[below] at (c) {$c$};
	    \node[below right] at (i) {$j$};
	    \node[below right] at (j) {$k$};
	    \node[below right] at (k) {$\ell$};
	    \fill (c) circle(2pt);
	    \fill (i) circle(2pt);
	    \fill (j) circle(2pt);
	    \fill (k) circle(2pt);

	    \draw[thick] (ia) -- (i) -- (ib);
	    \draw[thick] (ja) -- (j) -- (jb);
	    \draw[thick] (ka) -- (k) -- (kb);

	    \fill (ia) circle(2pt);
	    \fill (ib) circle(2pt);
	    \fill (ja) circle(2pt);
	    \fill (jb) circle(2pt);
	    \fill (ka) circle(2pt);
	    \fill (kb) circle(2pt);
	    
	    \node at (0,-2) (label){$G_1$};
	\end{tikzpicture}
	\end{minipage}
	\begin{minipage}{0.3\textwidth}
	\centering
	\begin{tikzpicture}[scale=1]
	\clip (-2,-2.5) rectangle (2,2.2);
	    \coordinate (c) at (0,0);
	    \coordinate (i) at (0,1);
	    \coordinate (j) at (-0.866,-0.5);
	    \coordinate (k) at (0.866,-0.5);
	    \coordinate (ja) at (-1.866,-0.5);
	    \coordinate (ka) at (1.866,-0.5);
	    \coordinate (jk) at (0,-1);
	    \coordinate (ia) at (0.866,1.5);
	    \coordinate (ib) at (-0.866,1.5);
	    
	    \draw[ultra thick] (c) -- (i);
	    \draw[ultra thick] (c) -- (j);
	    \draw[ultra thick] (c) -- (k);
	    \node[below] at (c) {$c$};
	    \node[below right] at (i) {$j$};
	    \node[below] at (j) {$k$};
	    \node[below] at (k) {$\ell$};
	    \fill (c) circle(2pt);
	    \fill (i) circle(2pt);
	    \fill (j) circle(2pt);
	    \fill (k) circle(2pt);

        \draw[thick] (j) -- (jk) -- (k);
        \fill (jk) circle(2pt);
        \draw[thick] (ia) -- (i) -- (ib);
        \fill (ia) circle(2pt);
	    \fill (ib) circle(2pt);
	    \draw[thick] (j) -- (ja);
        \fill (ja) circle(2pt);
        \draw[thick] (k) -- (ka);
        \fill (ka) circle(2pt);
        
	    \node at (0,-2) (label){$G_2$};
	\end{tikzpicture}
	\end{minipage}
	\begin{minipage}{0.3\textwidth}
	\centering
	\begin{tikzpicture}[scale=1]
	\clip (-2,-2.5) rectangle (2,2.2);
	    \coordinate (c) at (0,0);
	    \coordinate (i) at (0,1);
	    \coordinate (ia) at (0,2);
	    \coordinate (j) at (-0.866,-0.5);
	    \coordinate (ja) at (-1.732,-1);
	    \coordinate (k) at (0.866,-0.5);
	    \coordinate (ka) at (1.732,-1);
	    \coordinate (cp) at (0.866,0.2);
	    
	    \draw[ultra thick] (c) -- (i);
	    \draw[ultra thick] (c) -- (j);
	    \draw[ultra thick] (c) -- (k);
	    \node[above left] at (c) {$c$};
	    \node[left] at (i) {$j$};
	    \node[below] at (j) {$k$};
	    \node[below] at (k) {$\ell$};
	    \fill (c) circle(2pt);
	    \fill (i) circle(2pt);
	    \fill (j) circle(2pt);
	    \fill (k) circle(2pt);triangle

        \draw[thick] (cp) -- (i);
	    \draw[thick] (cp) -- (j);
	    \draw[thick] (cp) -- (k);
	    \fill (cp) circle(2pt);
	    
	    \draw[thick] (ja) -- (j);
	    \fill (ja) circle(2pt);
	    \draw[thick] (ia) -- (i);
	    \fill (ia) circle(2pt);environment
	    \draw[thick] (ka) -- (k);
	    \fill (ka) circle(2pt);
	    \node at (0,-2) (label){$G_3$};
	\end{tikzpicture}
	\end{minipage}
	\\
	\begin{minipage}{0.3\textwidth}
	\centering
	\begin{tikzpicture}[scale=1]
        \clip (-2,-2.5) rectangle (2,2.2);
	    \coordinate (c) at (0,0);
	    \coordinate (i) at (0,1);
	    \coordinate (ia) at (0.866,1.5);
	    \coordinate (ib) at (-0.866,1.5);
	    \coordinate (j) at (-0.866,-0.5);
	    \coordinate (k) at (0.866,-0.5);
	    \coordinate (jka) at (-0.366,-1.366);
	    \coordinate (jkb) at (0.366,-1.366);

	    \draw[ultra thick] (c) -- (i);
	    \draw[ultra thick] (c) -- (j);
	    \draw[ultra thick] (c) -- (k);
	    \node[below] at (c) {$c$};
	    \node[below right] at (i) {$j$};
	    \node[below left] at (j) {$k$};
	    \node[below right] at (k) {$\ell$};
	    \fill (c) circle(2pt);
	    \fill (i) circle(2pt);
	    \fill (j) circle(2pt);
	    \fill (k) circle(2pt);

	    \draw[thick] (ia) -- (i) -- (ib);

	    \fill (ia) circle(2pt);
	    \fill (ib) circle(2pt);

	    \fill (jka) circle(2pt);
	    \fill (jkb) circle(2pt);
	    
	    \draw[thick] (j) -- (jka) -- (k);
	    \draw[thick] (j) -- (jkb) -- (k);
	    
	    \node at (0,-2) (label){$G_4$};
	\end{tikzpicture}
	\end{minipage}
	\begin{minipage}{0.3\textwidth}
	\centering
	\begin{tikzpicture}[scale=1]
	\clip (-2,-2.5) rectangle (2,2.2);
	    \coordinate (c) at (0,0);
	    \coordinate (i) at (0,1);
	    \coordinate (j) at (-0.866,-0.5);
	    \coordinate (ja) at (-1.732,-1);
	    \coordinate (k) at (0.866,-0.5);
	    \coordinate (ka) at (1.732,-1);
	    \coordinate (ij) at (-0.866,0.5);
	    \coordinate (ik) at (0.866,0.5);
	   
	    \draw[ultra thick] (c) -- (i);
	    \draw[ultra thick] (c) -- (j);
	    \draw[ultra thick] (c) -- (k);
	    \node[below] at (c) {$c$};
	    \node[above] at (i) {$j$};
	    \node[below] at (j) {$k$};
	    \node[below] at (k) {$\ell$};
	    \fill (c) circle(2pt);
	    \fill (i) circle(2pt);
	    \fill (j) circle(2pt);
	    \fill (k) circle(2pt);triangle

	    \draw[thick] (ja) -- (j);
	    \fill (ja) circle(2pt);
	    \draw[thick] (ka) -- (k);
	    \fill (ka) circle(2pt);
	    \node at (0,-2) (label){$G_5$};
	    
	    \fill (ij) circle(2pt);
	    \fill (ik) circle(2pt);
	    \draw[thick] (j) -- (ij) -- (i) -- (ik) -- (k);
	\end{tikzpicture}
	\end{minipage}
	\begin{minipage}{0.3\textwidth}
	\centering
	\begin{tikzpicture}[scale=1]
	\clip (-2,-2.5) rectangle (2,2.2);
	    \coordinate (c) at (0,0);
	    \coordinate (i) at (0,1);
	    \coordinate (ia) at (0,2);
	    \coordinate (j) at (-0.866,-0.5);
	    \coordinate (k) at (0.866,-0.5);
	    \coordinate (cp) at (0.866,0.2);
	    \coordinate (jk) at (0,-1);
	    
	    \draw[ultra thick] (c) -- (i);
	    \draw[ultra thick] (c) -- (j);
	    \draw[ultra thick] (c) -- (k);
	    \node[above left] at (c) {$c$};
	    \node[left] at (i) {$j$};
	    \node[below] at (j) {$k$};
	    \node[below] at (k) {$\ell$};
	    \fill (c) circle(2pt);
	    \fill (i) circle(2pt);
	    \fill (j) circle(2pt);
	    \fill (k) circle(2pt);

        \draw[thick] (cp) -- (i);
	    \draw[thick] (cp) -- (j);
	    \draw[thick] (cp) -- (k);
	    \fill (cp) circle(2pt);
	    
	    \fill (jk) circle(2pt);
	    \draw[thick] (j)-- (jk) -- (k);
	    \draw[thick] (ia) -- (i);
	    \fill (ia) circle(2pt);

	    \node at (0,-2) (label){$G_6$};
	\end{tikzpicture}
	\end{minipage}
	\\
	\begin{minipage}{0.4\textwidth}
	\centering
		\begin{tikzpicture}[scale=1]
	\clip (-2,-2.5) rectangle (2,2.2);
	    \coordinate (c) at (0,0);
	    \coordinate (i) at (0,1);
	    \coordinate (j) at (-0.866,-0.5);
	    \coordinate (k) at (0.866,-0.5);
	    \coordinate (cp) at (0.866,0.2);
	    \coordinate (cpp) at (-0.866,0.2);
	    \coordinate (jk) at (0,-1);
	    
	    \draw[ultra thick] (c) -- (i);
	    \draw[ultra thick] (c) -- (j);
	    \draw[ultra thick] (c) -- (k);
	    \node[below=5pt] at (c) {$c$};
	    \node[above] at (i) {$j$};
	    \node[below] at (j) {$k$};
	    \node[below] at (k) {$\ell$};
	    \fill (c) circle(2pt);
	    \fill (i) circle(2pt);
	    \fill (j) circle(2pt);
	    \fill (k) circle(2pt);
G
        \draw[thick] (cp) -- (i);
	    \draw[thick] (cp) -- (j);
	    \draw[thick] (cp) -- (k);
	    \fill (cp) circle(2pt);
	    
	    \draw[thick] (cpp) -- (i);
	    \draw[thick] (cpp) -- (j);
	    \draw[thick] (cpp) -- (k);
	    \fill (cpp) circle(2pt);
	    \node at (0,-2) (label){$G_7$};
	\end{tikzpicture}
	\end{minipage}
	\begin{minipage}{0.4\textwidth}
	\centering
		\begin{tikzpicture}[scale=1]
    	\clip (-3, -3) rectangle (3, 3);
	    \coordinate (c) at (0,0);
	    \coordinate (i) at (0,1);
	    \coordinate (j) at (-0.866,-0.5);
	    \coordinate (k) at (0.866,-0.5);
	    \coordinate (ij) at (-0.866,0.5);
	    \coordinate (ik) at (0.866,0.5);
	    \coordinate (jk) at (0,-1);
	   
	    \draw[ultra thick] (c) -- (i);
	    \draw[ultra thick] (c) -- (j);
	    \draw[ultra thick] (c) -- (k);
	    \node[below] at (c) {$c$};
	    \node[above] at (i) {$j$};
	    \node[below] at (j) {$k$};
	    \node[below] at (k) {$\ell$};
	    \fill (c) circle(2pt);5
	    \fill (i) circle(2pt);
	    \fill (j) circle(2pt);
	    \fill (k) circle(2pt);triangle

	    \fill (ij) circle(2pt);
	    \fill (ik) circle(2pt);
	    \fill (jk) circle(2pt);
	    \draw[thick] (j) -- (ij) -- (i) -- (ik) -- (k) -- (jk) -- (j);
	    \node at (0,-2) (label){$G_8$};
	\end{tikzpicture}
	\end{minipage}
	\end{center}
	\caption[Triangle-free star environments, $p = 1$]{Triangle-free star environments in $\mathcal{E}^{(1)} \bigg( \triStarMinigraph \bigg)$ as described in Section~\ref{sec:localityuniformity}.}
	\label{fig:trifreestarp1Environments}
\end{figure}

\begin{figure}[!ht]
\begin{center}
	\begin{minipage}{0.45\textwidth}
	\centering
	\begin{tikzpicture}[scale=1]
        \clip (-3, -3) rectangle (3, 3);
	    \coordinate (j) at (-1,0);
	    \coordinate (i) at (0,0);
	    \coordinate (k) at (1,0);
	    \coordinate (ia) at (0,1);
	    \coordinate (ka) at (1.5,0.866);
	    \coordinate (kb) at (1.5,-0.866);
	    \coordinate (ja) at (-1.5,0.866);
	    \coordinate (jb) at (-1.5,-0.866);
	    
	    \draw[ultra thick] (j) -- (i) -- (k);
	    \node[below=3pt] at (i) {$c$};
	    \node[below right] at (j) {$j$};
	    \node[below left] at (k) {$k$};
	    \fill (i) circle(2pt);
	    \fill (j) circle(2pt);
	    \fill (k) circle(2pt);
	    
	    \fill (ia) circle(2pt);
	    \fill (ka) circle(2pt);
	    \fill (kb) circle(2pt);
	    \fill (ja) circle(2pt);
	    \fill (jb) circle(2pt);
	    
	    \draw[thick] (i) -- (ia);
	    \draw[thick] (ka) -- (k) -- (kb);
	    \draw[thick] (ja) -- (j) -- (jb);
	    
	    \node at (0,-2.5) (label){$T^{(1)} \big( \tripletMinigraph \big)$};
	\end{tikzpicture}
	\end{minipage}
	\begin{minipage}{0.45\textwidth}
	\centering
	\begin{tikzpicture}[scale=1]
        \clip (-3,-3) rectangle (3,3);
	    \coordinate (j) at (-1,0);
	    \coordinate (i) at (0,0);
	    \coordinate (k) at (1,0);
	    \coordinate (ia) at (0,1);
	    \coordinate (ja) at (-1.5,0.866);
	    \coordinate (jb) at (-1.5,-0.866);
	    \coordinate (ka) at (1.5,0.866);
	    \coordinate (kb) at (1.5,-0.866);
	    
	    \coordinate (jaa) at (-1.5, 1.866); 
	    \coordinate (jab) at (-2.5, 0.866); 
	    \coordinate (jba) at (-1.5, -1.866); 
	    \coordinate (jbb) at (-2.5, -0.866);

	    \coordinate (kaa) at (1.5, 1.866); 
	    \coordinate (kab) at (2.5, 0.866); 
	    \coordinate (kba) at (1.5, -1.866); 
	    \coordinate (kbb) at (2.5, -0.866); 
	    
	    \coordinate (iaa) at (-0.707, 1.707); 
	    \coordinate (iab) at (0.707, 1.707); 
	    
	    \draw[ultra thick] (j) -- (i) -- (k);
	    \node[below=3pt] at (i) {$c$};
	    \node[below right] at (j) {$j$};
	    \node[below left] at (k) {$k$};
	    \fill (i) circle(2pt);
	    \fill (j) circle(2pt);
	    \fill (k) circle(2pt);
	    
	    \fill (ia) circle(2pt);
	    \fill (ka) circle(2pt);
	    \fill (kb) circle(2pt);
	    \fill (ja) circle(2pt);
	    \fill (jb) circle(2pt);
	    
	    \fill (jaa) circle(2pt);
	    \fill (jab) circle(2pt);
	    \fill (kaa) circle(2pt);
	    \fill (kab) circle(2pt);
	    \fill (jba) circle(2pt);
	    \fill (jbb) circle(2pt);
	    \fill (kba) circle(2pt);
	    \fill (kbb) circle(2pt);
	    \fill (iaa) circle(2pt);
	    \fill (iab) circle(2pt);
	    
	    \draw[thick] (i) -- (ia);
	    \draw[thick] (ka) -- (k) -- (kb);
	    \draw[thick] (ja) -- (j) -- (jb);
	    \draw[thick] (jaa) -- (ja) -- (jab);
	    \draw[thick] (kaa) -- (ka) -- (kab);
	    \draw[thick] (jba) -- (jb) -- (jbb);
	    \draw[thick] (kba) -- (kb) -- (kbb);
	    \draw[thick] (iaa) -- (ia) -- (iab);
	    
	    \node at (0,-2.5) (label){$T^{(2)} \big( \tripletMinigraph \big)$};
	\end{tikzpicture}
	\end{minipage}
	\end{center}
	\caption[Triplet trees]{The triplet trees $T^{(1)} \big( \tripletMinigraph \big)$ and $T^{(2)} \big( \tripletMinigraph \big)$ constructed as described in Section~\ref{sec:localityuniformity}.}
	\label{fig:triplettrees}
\end{figure}

\begin{figure}[!ht]
\begin{center}
	\begin{minipage}{0.45\textwidth}
	\centering
	\begin{tikzpicture}[scale=1]
        \clip (-3, -3) rectangle (3, 3);
	    \coordinate (c) at (0,0);
	    \coordinate (i) at (0,1);
	    \coordinate (k) at (0.866,-0.5);
	    \coordinate (ia) at (0.866,1.5);
	    \coordinate (ib) at (-0.866,1.5);
	    \coordinate (j) at (-0.866,-0.5);
	    \coordinate (ja) at (-1.732,0);
	    \coordinate (jb) at (-0.866,-1.5);
	    \coordinate (ka) at (1.732,0);
	    \coordinate (kb) at (0.866,-1.5);

	    \draw[ultra thick] (c) -- (i);
	    \draw[ultra thick] (c) -- (j);
	    \draw[ultra thick] (c) -- (k);
	    \node[below] at (c) {$c$};
	    \node[below right] at (i) {$j$};
	    \node[below right] at (j) {$k$};
	    \node[below right] at (k) {$\ell$};
	    \fill (c) circle(2pt);
	    \fill (i) circle(2pt);
	    \fill (j) circle(2pt);
	    \fill (k) circle(2pt);

	    \draw[thick] (ia) -- (i) -- (ib);
	    \draw[thick] (ja) -- (j) -- (jb);
	    \draw[thick] (ka) -- (k) -- (kb);

	    \fill (ia) circle(2pt);
	    \fill (ib) circle(2pt);
	    \fill (ja) circle(2pt);
	    \fill (jb) circle(2pt);
	    \fill (ka) circle(2pt);
	    \fill (kb) circle(2pt);
	    
	    \node at (0,-2.5) (label){$T^{(1)} \bigg( \triStarMinigraph \bigg)$};
	\end{tikzpicture}
	\end{minipage}
	\begin{minipage}{0.45\textwidth}
	\centering
	\begin{tikzpicture}[scale=1]
        \clip (-3, -3) rectangle (3, 3);
	    \coordinate (c) at (0,0);
	    \coordinate (j) at (0,1);
	    \coordinate (l) at (0.866,-0.5);
	    \coordinate (k) at (-0.866,-0.5);
	    \coordinate (ja) at (0.866,1.5);
	    \coordinate (jb) at (-0.866,1.5);
	    \coordinate (ka) at (-1.732,0);
	    \coordinate (kb) at (-0.866,-1.5);
	    \coordinate (la) at (1.732,0);
	    \coordinate (lb) at (0.866,-1.5);
	    
	    \coordinate (jaa) at (-1.3, 1.25);
	    \coordinate (jab) at (-0.866, 2);
	    \coordinate (jba) at (1.3, 1.25);
	    \coordinate (jbb) at (0.866, 2);
	    \coordinate (kaa) at (-1.732, 0.5);
	    \coordinate (kab) at (-2.165, -0.25);
	    \coordinate (kba) at (-0.433, -1.732);
	    \coordinate (kbb) at (-1.3, -1.732);
	    \coordinate (laa) at (1.732, 0.5);
	    \coordinate (lab) at (2.165, -0.25);
	    \coordinate (lba) at (0.433, -1.732);
	    \coordinate (lbb) at (1.3, -1.732);

	    \draw[ultra thick] (c) -- (j);
	    \draw[ultra thick] (c) -- (k);
	    \draw[ultra thick] (c) -- (l);
	    \node[below] at (c) {$c$};
	    \node[below right] at (j) {$j$};
	    \node[below right] at (k) {$k$};
	    \node[below right] at (l) {$\ell$};
	    \fill (c) circle(2pt);
	    \fill (j) circle(2pt);
	    \fill (k) circle(2pt);
	    \fill (l) circle(2pt);
	    
	    \fill (ja) circle(2pt);
	    \fill (jb) circle(2pt);
	    \fill (ka) circle(2pt);
	    \fill (kb) circle(2pt);
	    \fill (la) circle(2pt);
	    \fill (lb) circle(2pt);
	    
	    \fill (jaa) circle(2pt); 
	    \fill (jab) circle(2pt); 
	    \fill (jba) circle(2pt); 
	    \fill (jbb) circle(2pt); 
	    \fill (kaa) circle(2pt); 
	    \fill (kab) circle(2pt); 
	    \fill (kba) circle(2pt); 
	    \fill (kbb) circle(2pt); 
	    \fill (laa) circle(2pt); 
	    \fill (lab) circle(2pt); 
	    \fill (lba) circle(2pt); 
	    \fill (lbb) circle(2pt); 

	    \draw[thick] (ja) -- (j) -- (jb);
	    \draw[thick] (ka) -- (k) -- (kb);
	    \draw[thick] (la) -- (l) -- (lb);
	    
	    \draw[thick] (jaa) -- (jb) -- (jab); 
	    \draw[thick] (jba) -- (ja) -- (jbb); 
	    \draw[thick] (kaa) -- (ka) -- (kab); 
	    \draw[thick] (kba) -- (kb) -- (kbb); 
	    \draw[thick] (laa) -- (la) -- (lab); 
	    \draw[thick] (lba) -- (lb) -- (lbb);

	    \node at (0,-2.5) (label){$T^{(2)} \bigg( \triStarMinigraph \bigg)$};
	\end{tikzpicture}
	\end{minipage}
	\end{center}
	\caption[Star trees]{The star trees $T^{(1)} \bigg( \triStarMinigraph \bigg)$ and $T^{(2)} \bigg( \triStarMinigraph \bigg)$ constructed as described in Section~\ref{sec:localityuniformity}.}
	\label{fig:startrees}
\end{figure}

\clearpage
\begin{landscape}
\section{Witness angles}

\begin{figure}[!ht]
\begin{center}
\begin{tabular}{ c | c c c c c c | c c c c c c }
Method & \multicolumn{6}{c|}{$\beta$} & \multicolumn{6}{c}{$\gamma$}\\ \hline
\tfklqp{2} & $0.99225$ & $3.46308$ & & & & & $5.78009$ & $2.25304$ & & & & \\
\thlzqp{2} & $0.98705$ & $3.47167$ & & & & & $5.77664$ & $2.25962$ & & & & \\
\hline
\tfklqp{3} & $0.62112$ & $0.48905$ & $0.26477$ & & & & $0.42728$ & $0.79596$ & $0.92620$ & & &\\
\thlzqp{3} & $0.62519$ & $0.49754$ & $0.27393$ & & & & $0.42808$ & $0.79569$ & $0.92077$ & & &\\
\hline
\qaoap{4}& $0.59956$ & $0.43434$ & $0.29676$ & $0.15904$ & & & $0.40875$ &$ 0.78057$ &$ 0.98804$ & $0.15691$ & &\\
\tfklqp{4} & $0.63219$ & $2.09215$ & $0.42150$ & $0.22286$ & & & $0.38433$ & $0.72509$ & $0.83266$	 & $0.94350$ & \\
\thlzqp{4} & $0.63516$ & $0.52634$ & $0.43047$ &	$0.23058$ & & & $0.38478$ & $0.72269$ & $0.82767$ & $0.93461$ & \\ \hline
\qaoap{5} & $0.63167$ & $0.52253$ & $1.96094$ & $0.27599$ & $0.14930$ & & $0.35924$ & $0.70609$ & $0.82209$ & $1.00420$ & $1.15394$ \\
\tfklqp{5} & $0.64008$ & $0.54030$ & $0.45437$ &	$0.34000$ & $0.18710$ & & $0.35582$ & $0.68736$ & $0.78042$ & $0.87482$ & $0.99556$ \\
\thlzqp{5} & $0.64349$ & $0.54679$ & $0.46687$ & $0.38838$ & $0.19975$ & & $0.35349$ & $0.68144$ & $0.76945$ & $0.85500$ & $0.96997$\\
\hline
\qaoap{6} & $0.63589$ & $0.53443$ & $0.46334$ & $0.35999$ & $0.25858$ & $0.13885$ & $0.33137$ & $0.64558$ & $0.73165$ & $0.83696$ & $1.01019$ & $1.12724$ \\ 
\tfklqp{6} & $0.64369$ & $0.54870$ & $0.47903$ & $0.40547$ & $1.88825$ & $0.16000$ & $0.33434$ & $0.64986$ & $0.73024$ & $0.81681$ & $0.93262$ & $1.04923$ \\
\thlzqp{6} & $0.64622$ & $0.55243$ & $0.48572$ & $0.42258$ & $1.91095$ & $0.17045$ & $0.33265$ & $0.64669$ & $0.72479$ & $0.80326$ & $0.90278$ & $1.01496$
\end{tabular}
\end{center}
\caption{Witness angles $\beta=\beta_1,\beta_2,\ldots$ and $\gamma=\gamma_1,\gamma_2,\ldots$ certifying the approximation ratios of \qaoap{p}, \tfklqp{p} and \thlzqp{p} on graphs of girth greater than $2p + 2$ for $p \in \{2, \ldots, 6\}$.}
\label{fig:angles}
\end{figure}
\end{landscape}
\enlargethispage{10\baselineskip}

\end{document}